\documentclass[journal]{IEEEtran}

\usepackage[utf8]{inputenc}
\usepackage{color}
\usepackage{xcolor}
\usepackage{array}
\usepackage{verbatim}
\usepackage{float}
\usepackage{amsmath}
\usepackage{amsthm}
\usepackage{amssymb}
\usepackage{graphicx}
\usepackage{longtable}
\usepackage{multirow}
\usepackage[unicode=true,
bookmarks=false,
breaklinks=false,pdfborder={0 0 1},colorlinks=false]
{hyperref}
\hypersetup{
	colorlinks,bookmarksopen,bookmarksnumbered,citecolor=blue,urlcolor=blue}
\usepackage{cite}

\usepackage{lipsum}
\usepackage{mathtools}
\usepackage{cuted}

\floatstyle{ruled}
\newfloat{algorithm}{tbp}{loa}
\providecommand{\algorithmname}{Algorithm}
\floatname{algorithm}{\protect\algorithmname}

\makeatletter
\let\oldforeign@language\foreign@language
\DeclareRobustCommand{\foreign@language}[1]{%
	\lowercase{\oldforeign@language{#1}}}

\let\oldforeign@language\foreign@language
\DeclareRobustCommand{\foreign@language}[1]{%
	\lowercase{\oldforeign@language{#1}}}

\ifCLASSINFOpdf
\else
\fi

\newcommand{\MYfooter}{\smash{
		\hfil\parbox[t][\height][t]{\textwidth}{\centering
			\thepage}\hfil\hbox{}}}

\def\ps@IEEEtitlepagestyle{%
	\def\@oddhead{\parbox[t][\height][t]{\textwidth}{\centering \scriptsize
			Personal use of this material is permitted. Permission from the author(s) and/or copyright holder(s), must be obtained for all other uses. Please contact us and provide details if you believe this document breaches copyrights.\\
			\noindent\makebox[\linewidth]{}
		}\hfil\hbox{}}%
	\def\@evenhead{\scriptsize\thepage \hfil \leftmark\mbox{}}%
	\def\@oddfoot{\parbox[t][\height][l]{\textwidth}{
			\vspace{-20pt}{\rule{\textwidth}{0.4pt}}\\ \footnotesize\underline{To cite this article:}
			{\bf{\footnotesize\textcolor{red}{H. A. Hashim, M. Abouheaf, and M. A. Abido "Geometric Stochastic Filter with Guaranteed Performance for Autonomous Navigation based on IMU and Feature Sensor Fusion," Control Engineering Practice, vol. 116, pp. 104926, 2021.}}} doi: \href{https://doi.org/10.1016/j.conengprac.2021.104926}{10.1016/j.conengprac.2021.104926}\\
			\noindent\makebox[\linewidth]
		}\hfil\hbox{}}%
	\def\@evenfoot{\MYfooter}}

\makeatother
\newtheorem{defn}{Definition}

\newtheorem{lem}{Lemma}

\newtheorem{thm}{Theorem}
\newtheorem{rem}{Remark}

\newtheorem{assum}{Assumption}

\begin{document}
	\bstctlcite{IEEEexample:BSTcontrol}

	\title{Geometric Stochastic Filter with Guaranteed Performance for Autonomous Navigation based on IMU and Feature Sensor Fusion}

\author{Hashim A. Hashim, Mohammed Abouheaf, and Mohammad A. Abido
	\thanks{This work was supported in part by Thompson Rivers University Internal research fund \# 102315.}
	\thanks{Corresponding author, H. A. Hashim is with the Department of Engineering and Applied Science, Thompson Rivers University, Kamloops, British Columbia, Canada, V2C-0C8, e-mail: hhashim@ieee.org, M. Abouheaf is with College of Technology, Architecture \& Applied Engineering, Bowling Green State University, Bowling Green, 43402, OH, USA, and M. A. Abido is with Electrical Engineering Department, King Fahd University of Petroleum and Minerals, Dhahran 31261, Saudi Arabia.
}
}



\maketitle

\begin{abstract}
This paper concerns the estimation problem of attitude, position,
and linear velocity of a rigid-body autonomously navigating with six
degrees of freedom (6 DoF). The navigation dynamics are highly nonlinear
and are modeled on the matrix Lie group of the extended Special Euclidean
Group $\mathbb{SE}_{2}(3)$. A computationally cheap geometric nonlinear
stochastic navigation filter is proposed on $\mathbb{SE}_{2}(3)$
with guaranteed transient and steady-state performance. The proposed
filter operates based on a fusion of sensor measurements collected
by a low-cost inertial measurement unit (IMU) and features (obtained
by a vision unit). The closed loop error signals are guaranteed to
be almost semi-globally uniformly ultimately bounded in the mean square
from almost any initial condition. The equivalent quaternion representation
is included in the Appendix. The filter is proposed in continuous
form, and its discrete form is tested on a real-world dataset of measurements
collected by a quadrotor navigating in three dimensional (3D) space.
\end{abstract}


\IEEEpeerreviewmaketitle{}

\rule{0.47\textwidth}{1pt}\\
Please visit the last page for Bibtex citation and video URL of the experiement\\
\vspace{-1pt}
\rule{0.49\textwidth}{1pt}

\section{Introduction}

\IEEEPARstart{R}{obust} and accurate navigation solutions for autonomous vehicles are
essential \cite{barrau2016invariant,chebly2019coupled,hashim2020SE3Stochastic,hua2018attitude,hashim2020TITS_SLAM,bucsoniu2020learning}.
Indoor and outdoor applications, such as household cleaning devices,
pipelines, terrain mapping, reef monitoring, exemplify situations
when GPS might be unreliable and only low-cost measurement units (\textit{e.g.},
inertial measurement unit (IMU)) might be available. In such a case
GPS-independent navigation solutions are indispensable. A typical
low-cost IMU module is composed of an accelerometer and a gyroscope
which provide measurements of rigid-body's acceleration and angular
velocity, respectively. In the absence of GPS, a cost-effective autonomous
vehicle requires navigation solutions that rely on low-cost IMU and
feature measurements collected by a vision unit. Consequently, linear
velocity cannot be measured and its integration is impracticable due
to the unbounded error growth resulting from sensor noise and bias
\cite{woodman2007introduction}. Hence, autonomous navigation in space
requires estimation of orientation (known as attitude), position,
and linear velocity. Attitude can be successfully extracted given
a low-cost IMU module using Gaussian filters \cite{markley2003attitude,zamani2013minimum,choukroun2006novel},
nonlinear deterministic filters on the Special Orthogonal Group $\mathbb{SO}(3)$
\cite{lee2012exponential,mahony2008nonlinear,batista2012globally}
or nonlinear stochastic filters on $\mathbb{SO}(3)$ \cite{hashim2019SO3Wiley}.
The nonlinear deterministic filters in \cite{lee2012exponential,mahony2008nonlinear,hashim2019SE3Det,vasconcelos2010nonlinear,baldwin2009nonlinear}
have been proven to be almost globally asymptotically stable, while
nonlinear stochastic filters in \cite{hashim2019SO3Wiley} are guaranteed
to be almost globally asymptotically stable with a probability of
one in mean square.

An inertial vision unit composed of a stereo vision unit and an IMU
can be employed to extract rigid-body's pose - a combination of attitude
and position. Pose estimation is commonly approached using Gaussian
filters (Kalman-type filters) \cite{janabi2010kalman}. Nonetheless,
nonlinear deterministic filters \cite{vasconcelos2010nonlinear,baldwin2009nonlinear,hua2011observer}
and nonlinear stochastic filters \cite{hashim2020SE3Stochastic} evolved
on the Special Euclidean Group $\mathbb{SE}(3)$ have been deemed
more suitable. Nonlinear deterministic pose filters have been shown
to be almost globally asymptotically stable \cite{hashim2019SE3Det,vasconcelos2010nonlinear,baldwin2009nonlinear},
while nonlinear stochastic pose filters are guaranteed to be semi-globally
uniformly ultimately bounded in mean square \cite{hashim2020SE3Stochastic}.
It should be remarked that although the use of low-cost IMU and vision
units facilitates the development of cost-effective autonomous vehicles,
it requires robust solutions that accommodate for the measurement
uncertainties.

The true six degrees of freedom (6 DoF) vehicle navigation dynamics
are composed of attitude, position, and linear velocity dynamics modeled
on the Lie group of $\mathbb{SE}_{2}(3)$. The true dynamics are highly
nonlinear, and are neither right nor left invariant. Navigation problem
has been addressed using Kalman-type filters (KF) \cite{davari2016asynchronous,bijker2008kalman},
extended Kalman filters (EKF) \cite{anderson2012optimal}, unscented
Kalman filters (UKF) \cite{zhang2005navigation}, multiplicative extended
Kalman filter (MEKF) \cite{leishman2015multiplicative}, and particle
filters (PFs) \cite{zhao2014particle}. However, the navigation solutions
involving KF, EKF, UKF, MEKF, and PFs are not posed on $\mathbb{SE}_{2}(3)$.
In view of the nonlinearity of the navigation dynamics, several solutions
have been developed on the Lie group of $\mathbb{SE}_{2}(3)$, including
invariant extended Kalman filter (IEKF) \cite{barrau2016invariant},
a Riccati observer \cite{hua2018riccati}, and a nonlinear stochastic
observer \cite{hashim2021ACC}. Stochastic filters have been useful
in several applications, for instance distributed delays \cite{wei2021input}
and simultaneous localization and mapping \cite{Hashim2021AESCTE}.
The main shortcomings of these solutions are 1) the disregard for
the IMU measurement noise and 2) the lack of transient and steady-state
performance measures. Guaranteed convergence can be achieved by applying
a prescribed performance function (PPF) \cite{bechlioulis2008robust}.
PPF manipulates the error to initiate within a known large set and
to decay systematically ultimately settling within a known small set
\cite{bechlioulis2008robust}.

To summarize, the true navigation dynamics are 1) highly nonlinear,
2) posed on $\mathbb{SE}_{2}(3)$, and 3) reliant on angular velocity
and acceleration. Addressing all of the above-mentioned navigation
dynamics characteristics, 
\begin{enumerate}
	\item[1)] this paper introduces a geometric nonlinear stochastic filter for
	inertial navigation posed on $\mathbb{SE}_{2}(3)$ reliant on measurements
	supplied by a typical low-cost IMU module, namely angular velocity
	and acceleration measurements commonly corrupted by noise,
	\item[2)] the proposed stochastic filter is characterized by guaranteed transient
	and steady-state performance of the attitude and position error,
	\item[3)] the closed loop error signals are shown to be almost semi-globally
	uniformly ultimately bounded in the mean square,
	\item[4)] although the stochastic filter is developed in a continuous form,
	its discrete representation is included, and
	\item[5)] the computational inexpensiveness of the proposed stochastic filter
	is demonstrated using a real-world three dimensional (3D) quadrotor
	dataset tested at a low sampling rate.
\end{enumerate}
The proposed navigation filter is applicable to unmanned aerial vehicles
as well as ground robots.

The paper is composed of seven sections: Section \ref{sec:Preliminary-Material}
presents important notation and preliminaries. Section \ref{sec:SE3_Problem-Formulation}
discusses the true navigation problem in a stochastic sense, available
measurements, and error criteria. Section \ref{sec:Guaranteed-Performance}
reformulates the stochastic dynamics to follow systematic measures.
Section \ref{sec:Non-Nav-Observer1} proposes a novel nonlinear stochastic
filter. Section \ref{sec:SE3_Simulations} reveals experimental results.
Lastly, Section \ref{sec:SE3_Conclusion} concludes the work.

\section{Preliminary Material\label{sec:Preliminary-Material}}

A set of real numbers, nonnegative real numbers, and an $n$-by-$m$
real dimensional space are defined by $\mathbb{R}$, $\mathbb{R}_{+}$,
and $\mathbb{R}^{n\times m}$, respectively. The Euclidean norm of
$x\in\mathbb{R}^{n}$ is $||x||=\sqrt{x^{\top}x}$, while the Frobenius
norm of $M\in\mathbb{R}^{n\times m}$ is $||M||_{F}=\sqrt{{\rm Tr}\{MM^{*}\}}$
where $*$ denotes a conjugate transpose. $\mathbf{I}_{n}$ is an
$n$-by-$n$ identity matrix. For $M\in\mathbb{R}^{n\times n}$, the
set of eigenvalues is $\lambda(M)=\{\lambda_{1},\lambda_{2},\ldots,\lambda_{n}\}$
where $\overline{\lambda}_{M}=\overline{\lambda}(M)$ and $\underline{\lambda}_{M}=\underline{\lambda}(M)$
are the maximum and the minimum of $\lambda(M)$, respectively. $0_{n\times m}$
is a zero matrix and $1_{n\times m}$ is an $n$-by-$m$ matrix of
ones. $\mathbb{P}\{\cdot\}$ and $\mathbb{E}[\cdot]$ denote a probability
and an expected value of a component, respectively. The fixed inertial-
and body-frame are represented by $\left\{ \mathcal{I}\right\} $
and $\left\{ \mathcal{B}\right\} $, respectively. The vehicle's attitude
is denoted by $R\in\mathbb{SO}\left(3\right)$ with $\mathbb{SO}\left(3\right)$
defined by 
\[
\mathbb{SO}(3)=\{R\in\mathbb{R}^{3\times3}|RR^{\top}=R^{\top}R=\mathbf{I}_{3}\text{, }{\rm det}(R)=+1\}
\]
where ${\rm det}(\cdot)$ is a determinant. $\mathfrak{so}(3)$ is
the Lie algebra of $\mathbb{SO}(3)$ given by 
\[
\mathfrak{so}\left(3\right)=\{[x]_{\times}\in\mathbb{R}^{3\times3}|[x]_{\times}^{\top}=-[x]_{\times},x\in\mathbb{R}^{3}\}
\]
where $[x]_{\times}$is a skew symmetric matrix defined by
\begin{align*}
	\left[x\right]_{\times} & =\left[\begin{array}{ccc}
		0 & -x_{3} & x_{2}\\
		x_{3} & 0 & -x_{1}\\
		-x_{2} & x_{1} & 0
	\end{array}\right]\in\mathfrak{so}\left(3\right),\hspace{1em}x=\left[\begin{array}{c}
		x_{1}\\
		x_{2}\\
		x_{3}
	\end{array}\right]
\end{align*}
$\mathbf{vex}:\mathfrak{so}\left(3\right)\rightarrow\mathbb{R}^{3}$
describes the inverse mapping of $[\cdot]_{\times}$ such that $\mathbf{vex}([x]_{\times})=x,\forall x\in\mathbb{R}^{3}$.
$\boldsymbol{\mathcal{P}}_{a}$ denotes the anti-symmetric projection
on the $\mathfrak{so}\left(3\right)$ given by 
\[
\boldsymbol{\mathcal{P}}_{a}(M)=\frac{1}{2}(M-M^{\top})\in\mathfrak{so}\left(3\right),\forall M\in\mathbb{R}^{3\times3}
\]
Let $\boldsymbol{\Upsilon}=\mathbf{vex}\circ\boldsymbol{\mathcal{P}}_{a}$
denote the composition mapping such that 
\[
\boldsymbol{\Upsilon}(M)=\mathbf{vex}(\boldsymbol{\mathcal{P}}_{a}(M))\in\mathbb{R}^{3},\forall M\in\mathbb{R}^{3\times3}
\]
Consider $||R||_{{\rm I}}$ as the Euclidean distance of $R\in\mathbb{SO}\left(3\right)$
expressed by
\begin{equation}
||R||_{{\rm I}}=\frac{1}{4}{\rm Tr}\{\mathbf{I}_{3}-R\}\in\left[0,1\right]\label{eq:NAV_Ecul_Dist}
\end{equation}
It is worth noting that $-1\leq{\rm Tr}\{R\}\leq3$ and $||R||_{{\rm I}}=\frac{1}{8}||\mathbf{I}_{3}-R||_{F}^{2}$,
visit \cite{hashim2019SO3Wiley}. Let the attitude, position, and
linear velocity of a vehicle navigating in 3D space be $R\in\mathbb{SO}\left(3\right)$,
$P\in\mathbb{R}^{3}$, and $V\in\mathbb{R}^{3}$, respectively, for
all $R\in\{\mathcal{B}\}$ and $P,V\in\{\mathcal{I}\}$. The Special
Euclidean Group contains $R$ and $P$ defined by $\mathbb{SE}\left(3\right)=\mathbb{SO}\left(3\right)\times\mathbb{R}^{3}\subset\mathbb{R}^{4\times4}$,
visit \cite{hashim2020SE3Stochastic}. The extended form of $\mathbb{SE}\left(3\right)$
is $\mathbb{SE}_{2}\left(3\right)=\mathbb{SO}\left(3\right)\times\mathbb{R}^{3}\times\mathbb{R}^{3}\subset\mathbb{R}^{5\times5}$
defined by
\begin{align}
	\mathbb{SE}_{2}\left(3\right) & =\{\left.X\in\mathbb{R}^{5\times5}\right|R\in\mathbb{SO}\left(3\right),P,V\in\mathbb{R}^{3}\}\label{eq:NAV_SE2_3}\\
	X=\Psi(R & ,P,V)=\left[\begin{array}{ccc}
		R & P & V\\
		0_{1\times3} & 1 & 0\\
		0_{1\times3} & 0 & 1
	\end{array}\right]\in\mathbb{SE}_{2}\left(3\right)\label{eq:NAV_X}
\end{align}
where $X\in\mathbb{SE}_{2}\left(3\right)$ is a homogeneous navigation
matrix composed of rigid-body's attitude, position and linear velocity.
Define $\mathcal{U}_{m}=\mathfrak{so}\left(3\right)\times\mathbb{R}^{3}\times\mathbb{R}^{3}\times\mathbb{R}\subset\mathbb{R}^{5\times5}$
as follows:
\begin{align}
	\mathcal{U}_{m} & =\left\{ \left.u([\Omega\text{\ensuremath{]_{\times}}},V,a,\kappa)\right|[\Omega\text{\ensuremath{]_{\times}}}\in\mathfrak{so}(3),V,a\in\mathbb{R}^{3},\kappa\in\mathbb{R}\right\} \nonumber \\
	& u([\Omega\text{\ensuremath{]_{\times}}},V,a,\kappa)=\left[\begin{array}{ccc}
		[\Omega\text{\ensuremath{]_{\times}}} & V & a\\
		0_{1\times3} & 0 & 0\\
		0_{1\times3} & \kappa & 0
	\end{array}\right]\in\mathcal{U}_{m}\subset\mathbb{R}^{5\times5}\label{eq:NAV_u}
\end{align}
where $\Omega\in\mathbb{R}^{3}$, $V\in\mathbb{R}^{3}$, and $a\in\mathbb{R}^{3}$
denote the vehicle's true angular velocity, linear velocity, and apparent
acceleration comprised of all non-gravitational forces affecting the
vehicle, respectively, for all $\Omega,a\in\{\mathcal{B}\}$. 

\section{Problem Formulation\label{sec:SE3_Problem-Formulation}}

From \eqref{eq:NAV_X}, the true dynamics of a rigid-body navigating
in 3D space are given by
\begin{equation}
\begin{cases}
\dot{R} & =R\left[\Omega\right]_{\times}\\
\dot{P} & =V\\
\dot{V} & =Ra+\overrightarrow{\mathtt{g}}
\end{cases}\label{eq:NAV_Detailed_True_dot}
\end{equation}
where $R\in\mathbb{SO}\left(3\right)$, $P\in\mathbb{R}^{3}$, and
$V\in\mathbb{R}^{3}$ represent attitude, position, and linear velocity
of a vehicle navigating in 3D space, respectively, and $\overrightarrow{\mathtt{g}}$
is a gravity vector. Also, $\Omega\in\mathbb{R}^{3}$ denotes the
vehicle's true angular velocity, $V\in\mathbb{R}^{3}$ denotes the
vehicle's true linear velocity, and $a\in\mathbb{R}^{3}$ denotes
the apparent acceleration comprised of all non-gravitational forces
affecting the vehicle. Note that $R,\Omega,a\in\{\mathcal{B}\}$ and
$P,V\in\{\mathcal{I}\}$. Express \eqref{eq:NAV_Detailed_True_dot}
more compactly as
\begin{align}
	\dot{X}= & XU-\mathcal{\mathcal{G}}X\label{eq:NAV_True_dot}\\
	= & \left[\begin{array}{ccc}
		R & P & V\\
		0_{1\times3} & 1 & 0\\
		0_{1\times3} & 0 & 1
	\end{array}\right]\left[\begin{array}{ccc}
		\left[\Omega\right]_{\times} & 0_{3\times1} & a\\
		0_{1\times3} & 0 & 0\\
		0_{1\times3} & 1 & 0
	\end{array}\right]\nonumber \\
	& -\left[\begin{array}{ccc}
		0_{3\times3} & 0_{3\times1} & -\overrightarrow{\mathtt{g}}\\
		0_{1\times3} & 0 & 0\\
		0_{1\times3} & 1 & 0
	\end{array}\right]\left[\begin{array}{ccc}
		R & P & V\\
		0_{1\times3} & 1 & 0\\
		0_{1\times3} & 0 & 1
	\end{array}\right]\nonumber 
\end{align}
where $X\in\mathbb{SE}_{2}\left(3\right)$, $U=u([\Omega\text{\ensuremath{]_{\times}}},0_{3\times1},a,1)\in\mathcal{U}_{m}$,
and $\mathcal{\mathcal{G}}=u(0_{3\times3},0_{3\times1},-\overrightarrow{\mathtt{g}},1)\in\mathcal{U}_{m}$,
see \eqref{eq:NAV_u}. Note that $T_{X}\mathbb{SE}_{2}\left(3\right)\in\mathbb{R}^{5\times5}$
denotes the tangent space of $\mathbb{SE}_{2}\left(3\right)$ at point
$X$ where $\dot{X}:\mathbb{SE}_{2}\left(3\right)\times\mathcal{U}_{m}\rightarrow T_{X}\mathbb{SE}_{2}\left(3\right)$.
When low-cost sensors are used in a GPS-denied environment, all the
components of the navigation matrix $X$ become unknown. Figure \ref{fig:NAVIGATION}
presents an ample illustration of the navigation problem.

\begin{figure*}
	\centering{}\includegraphics[scale=0.45]{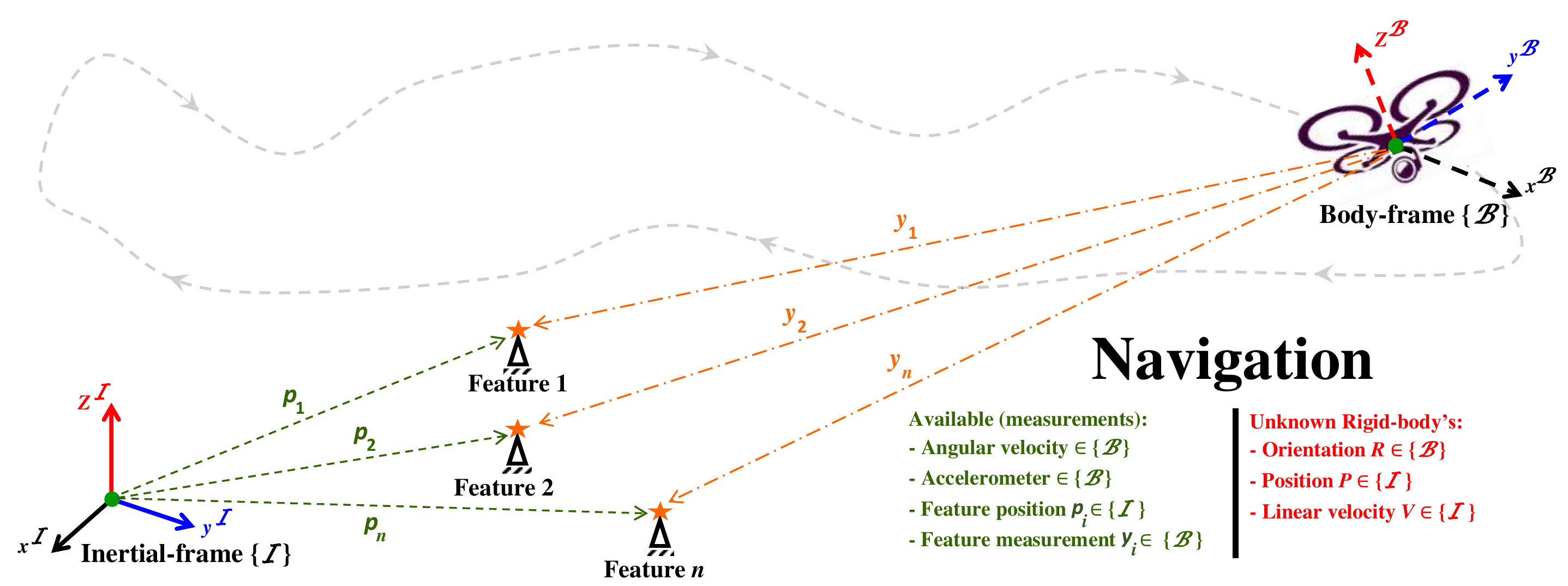}\caption{Navigation estimation task.}
	\label{fig:NAVIGATION}
\end{figure*}

The estimation can be performed using a set of sensor measurements.
The measurements of $\Omega$ and $a$ are as follows: 
\begin{equation}
\begin{cases}
\Omega_{m} & =\Omega+n_{\Omega}\in\mathbb{R}^{3}\\
a_{m} & =a+n_{a}\in\mathbb{R}^{3}
\end{cases}\label{eq:NAV_XVelcoity}
\end{equation}
where $n_{\Omega}$ and $n_{a}$ denote unknown bounded zero-mean
noise contaminating $\Omega$ and $a$, respectively, and $\mathbb{E}[n_{\Omega}]=\mathbb{E}[n_{a}]=0_{3\times1}$.
Measurements of $\Omega$ and $a$ can be easily obtained by a low-cost
IMU module. Due to the fact that the derivative of a Gaussian process
results in a Gaussian process, $n_{\Omega}=\mathcal{Q}d\beta_{\Omega}/dt$
and $n_{a}=\mathcal{Q}d\beta_{a}/dt$ as defined relative to the Brownian
motion process vector \cite{jazwinski2007stochastic} where $\mathcal{Q}\in\mathbb{R}^{3\times3}$
is an unknown positive time-variant diagonal matrix. The covariance
of $n_{\Omega}$ and $n_{a}$ is $\mathcal{Q}^{2}=\mathcal{Q}\mathcal{Q}^{\top}$.
Note that the Brownian motion process is characterized by \cite{jazwinski2007stochastic,deng2001stabilization,ito1984lectures,hashim2019SO3Wiley}:
\begin{align}
	\mathbb{P}\{\beta_{\Omega}(0)=\beta_{a}(0)=0\} & =1\nonumber \\
	\mathbb{E}[d\beta_{\Omega}/dt]=\mathbb{E}[d\beta_{a}/dt] & =0_{3\times1}\nonumber \\
	\mathbb{E}[\beta_{\Omega}]=\mathbb{E}[\beta_{a}] & =0_{3\times1}\label{eq:NAV_STC_Brownian}
\end{align}
Thus, in the light of \eqref{eq:NAV_Ecul_Dist}, the navigation dynamics
in \eqref{eq:NAV_Detailed_True_dot} can be expressed as a stochastic
differential equation:
\begin{equation}
\begin{cases}
d||R||_{{\rm I}} & =(1/2)\mathbf{vex}(\boldsymbol{\mathcal{P}}_{a}(R))^{\top}(\Omega_{m}dt-\mathcal{Q}d\beta_{\Omega})\\
dP & =Vdt\\
dV & =(Ra_{m}+\overrightarrow{\mathtt{g}})dt-R\mathcal{Q}d\beta_{a}
\end{cases}\label{eq:NAV_Detailed_True_STCH_dot}
\end{equation}
where $dR=R[\Omega_{m}]_{\times}dt-R[\mathcal{Q}d\beta_{\Omega}]_{\times}$
and ${\rm Tr}\{R[\Omega_{m}]_{\times}\}={\rm Tr}\{\boldsymbol{\mathcal{P}}_{a}(R)[\Omega_{m}]_{\times}\}=-2\mathbf{vex}(\boldsymbol{\mathcal{P}}_{a}(R))^{\top}\Omega_{m}$,
visit \cite{hashim2019SO3Wiley}. To achieve adaptive stabilization,
define
\begin{equation}
\sigma=[\sup_{t\geq0}\mathcal{Q}_{(1,1)},\sup_{t\geq0}\mathcal{Q}_{(2,2)},\sup_{t\geq0}\mathcal{Q}_{(3,3)}]^{\top}\label{eq:NAV_s}
\end{equation}
Let $n$ features be observed in the inertial-frame $p_{i}\in\left\{ \mathcal{I}\right\} $
and measured in the body-frame such that
\begin{align}
	\overline{y}_{i} & =X^{-1}\overline{p}_{i}+[(n_{i}^{y})^{\top},0,0]^{\top}\in\mathbb{R}^{5}\nonumber \\
	y_{i} & =R^{\top}(p_{i}-P)+n_{i}^{y}\in\mathbb{R}^{3}\label{eq:NAV_Vec_Landmark}
\end{align}
with $X^{-1}=\left[\begin{array}{ccc}
R^{\top} & -R^{\top}P & -R^{\top}V\\
0_{1\times3} & 1 & 0\\
0_{1\times3} & 0 & 1
\end{array}\right]$, $\overline{y}_{i}=[y_{i}^{\top},1,0]^{\top}$, $\overline{p}_{i}=[p_{i}^{\top},1,0]^{\top}$,
and $n_{i}^{y}$ being the noise associated with $y_{i}\in\left\{ \mathcal{B}\right\} $. 

\begin{assum}\label{Assum:NAV_1Landmark}The number of noncollinear
	features available for measurement is greater than or equal to three
	in order to define a plane.\end{assum}

\begin{defn}
	\label{def:NAV_SGUUB}\cite{ji2006adaptive} Consider the stochastic
	differential system in \eqref{eq:NAV_Detailed_True_STCH_dot} with
	$x=[||R||_{{\rm I}},P^{\top},V^{\top}]^{\top}$. $x(t)$ is almost
	semi-globally uniformly ultimately bounded if for a known set $\Sigma\in\mathbb{R}^{7}$
	and $x(t_{0})$, there exist a constant $\kappa>0$ and a time constant
	$\tau=\tau(\kappa,x(t_{0}))$ with $\mathbb{E}[||x(t_{0})||]<\kappa,\forall t>t_{0}+\tau$. 
\end{defn}
\begin{defn}
	\label{def:NAV_LV} Translate the stochastic dynamics in \eqref{eq:NAV_Detailed_True_STCH_dot}
	into a vector form such that 
	\[
	dx=fdt+h\overline{\mathcal{Q}}d\beta
	\]
	with $x=[||R||_{{\rm I}},P^{\top},V^{\top}]^{\top}\in\mathbb{R}^{7}$,
	$f=[(1/2)\mathbf{vex}(\boldsymbol{\mathcal{P}}_{a}(R))^{\top}\Omega_{m},V^{\top},(Ra_{m}+\overrightarrow{\mathtt{g}})^{\top}]^{\top}\in\mathbb{R}^{7}$,
	$h\in\mathbb{R}^{7\times9}$, $\overline{\mathcal{Q}}={\rm diag}(\mathcal{Q},\mathcal{Q},\mathcal{Q})\in\mathbb{R}^{9\times9}$,
	and $\beta=[\beta_{\Omega}^{\top},0_{3\times1}^{\top},\beta_{a}^{\top}]^{\top}\in\mathbb{R}^{9}$.
	Let $\mathbb{V}(x)$ be a twice differentiable function $\mathbb{V}(x)\in\mathcal{C}^{2}$.
	The differential operator $\mathcal{L}\mathbb{V}(x)$ is defined as
	follows: 
	\[
	\mathcal{L}\mathbb{V}(x)=\mathbb{V}_{x}^{\top}f+\frac{1}{2}{\rm Tr}\{h\overline{\mathcal{Q}}^{2}h^{\top}\mathbb{V}_{xx}\}
	\]
	where $\mathbb{V}_{x}=\partial\mathbb{V}/\partial x$ and $\mathbb{V}_{xx}=\partial^{2}\mathbb{V}/\partial x^{2}$.
\end{defn}
\begin{lem}
	\label{Lemm:NAV_deng}\cite{deng2001stabilization} Let the stochastic
	system in \eqref{eq:NAV_Detailed_True_STCH_dot} with $x=[||R||_{{\rm I}},P^{\top},V^{\top}]^{\top}$
	be expressed as $dx=fdt+h\overline{\mathcal{Q}}d\beta$ with $f\in\mathbb{R}^{7}$,
	$h\in\mathbb{R}^{7\times9}$, $\overline{\mathcal{Q}}\in\mathbb{R}^{9\times9}$,
	and $\beta=[\beta_{\Omega}^{\top},0_{3\times1}^{\top},\beta_{a}^{\top}]^{\top}\in\mathbb{R}^{9}$.
	Consider a twice differentiable potential function $\mathbb{V}(x)$
	where $\mathbb{V}:\mathbb{R}^{7}\rightarrow\mathbb{R}_{+}$. Let $\alpha_{1}(\cdot)$
	and $\alpha_{2}(\cdot)$ be class $\mathcal{K}_{\infty}$ functions,
	and let constants $\eta_{1}>0$ and $\eta_{2}\geq0$ such that
	\begin{align}
		& \hspace{1em}\alpha_{1}(x)\leq\mathbb{V}(x)\leq\alpha_{2}(x)\label{eq:NAV_Vfunction_Lyap}\\
		\mathcal{L}\mathbb{V}(x)= & (\partial\mathbb{V}/\partial x)^{\top}f+\frac{1}{2}{\rm Tr}\{h\overline{\mathcal{Q}}^{2}h^{\top}(\partial^{2}\mathbb{V}/\partial x^{2})\}\nonumber \\
		\leq & -\eta_{1}\mathbb{V}(x)+\eta_{2}\label{eq:NAV_dVfunction_Lyap}
	\end{align}
	Then for $x\in\mathbb{R}^{7}$, the stochastic dynamics in \eqref{eq:NAV_Detailed_True_STCH_dot}
	have almost a unique strong solution on $[0,\infty)$. Also, the solution
	$x$ is bounded in probability where
	\begin{equation}
	\mathbb{E}[\mathbb{V}(x)]\leq\mathbb{V}(x(0)){\rm exp}(-\eta_{1}t)+\eta_{2}/\eta_{1}\label{eq:NAV_EVfunction_Lyap}
	\end{equation}
	In addition, the inequality in \eqref{eq:NAV_EVfunction_Lyap} indicates
	that $x$ is semi-globally uniformly ultimately bounded%
	. 
\end{lem}

\subsection{Estimates, Error, and Measurements Setup\label{subsec:Navigation-Matrix}}

Define the estimate of $X\in\mathbb{SE}_{2}\left(3\right)$ in \eqref{eq:NAV_X}
as
\begin{equation}
\hat{X}=\Psi(\hat{R},\hat{P},\hat{V})=\left[\begin{array}{ccc}
\hat{R} & \hat{P} & \hat{V}\\
0_{1\times3} & 1 & 0\\
0_{1\times3} & 0 & 1
\end{array}\right]\in\mathbb{SE}_{2}\left(3\right)\label{eq:NAV_X_est}
\end{equation}
with $\hat{R}\in\mathbb{SO}\left(3\right)$, $\hat{P}\in\mathbb{R}^{3}$,
and $\hat{V}\in\mathbb{R}^{3}$ being the estimates of $R$, $P$,
and $V$, respectively. Let the estimation error between $X$ and
$\hat{X}$ be
\begin{align*}
	\tilde{X}=X\hat{X}^{-1} & =\left[\begin{array}{ccc}
		\tilde{R} & \tilde{P} & \tilde{V}\\
		0_{1\times3} & 1 & 0\\
		0_{1\times3} & 0 & 1
	\end{array}\right]\in\mathbb{SE}_{2}\left(3\right)
\end{align*}
where $\hat{X}^{-1}=\Psi(\hat{R}^{\top},-\hat{R}^{\top}\hat{P},-\hat{R}^{\top}\hat{V})$,
$\tilde{R}=R\hat{R}^{\top}$, $\tilde{P}=P-\tilde{R}\hat{P}$, and
$\tilde{V}=V-\tilde{R}\hat{V}$. The navigation estimation problem
aims to drive $X\rightarrow\hat{X}$ such that $\tilde{X}\rightarrow\mathbf{I}_{5}$,
$\tilde{R}\rightarrow\mathbf{I}_{3}$, $\tilde{P}\rightarrow0_{3\times1}$,
and $\tilde{V}\rightarrow0_{3\times1}$. Define the error
\begin{align}
	\overset{\circ}{\tilde{y}}_{i} & =\overline{p}_{i}-\tilde{X}^{-1}\overline{p}_{i}=\overline{p}_{i}-\hat{X}\overline{y}_{i}\nonumber \\
	& =[\underbrace{(p_{i}-\hat{R}y_{i}-\hat{P})^{\top}}_{\tilde{p}_{i}-\tilde{P}},0,0]^{\top}\label{eq:NAV_y_error}
\end{align}
where $\overset{\circ}{\tilde{y}}_{i}=[\tilde{y}_{i}^{\top},0,0]^{\top}$,
$\tilde{p}_{i}=\hat{p}_{i}-\tilde{R}p_{i}$, and $\tilde{P}=\hat{P}-\tilde{R}P$.
Let $s_{i}>0$ refer to the sensor confidence level of the $i$th
feature. Define the following components with respect to the available
vector measurements: $s_{T}=\sum_{i=1}^{n}s_{i}$, $p_{c}=\frac{1}{s_{T}}\sum_{i=1}^{n}s_{i}p_{i}$,
$M=\sum_{i=1}^{n}s_{i}(p_{i}-p_{c})(p_{i}-p_{c})^{\top}$ such that
$M=\sum_{i=1}^{n}s_{i}p_{i}p_{i}^{\top}-s_{T}p_{c}p_{c}^{\top}$,
$\sum_{i=1}^{n}s_{i}(p_{i}-p_{c})y_{i}^{\top}\hat{R}^{\top}=\sum_{i=1}^{n}s_{i}(p_{i}-p_{c})(p_{i}-P)^{\top}\tilde{R}=M\tilde{R}$,
and $\sum_{i=1}^{n}s_{i}\tilde{y}_{i}=\sum_{i=1}^{n}s_{i}(p_{i}-\hat{R}y_{i}-\hat{P})=s_{T}\tilde{R}^{\top}\tilde{P}_{\varepsilon}$
with $\tilde{P}_{\varepsilon}=\tilde{P}-(\mathbf{I}_{3}-\tilde{R})p_{c}$.
It can be deduced that $\tilde{R}\rightarrow\mathbf{I}_{3}$ indicates
that $\tilde{P}_{\varepsilon}\rightarrow\tilde{P}$. Let us sum up
the vector measurements:
\begin{equation}
\begin{cases}
p_{c} & =\frac{1}{s_{T}}\sum_{i=1}^{n}s_{i}p_{i},\hspace{1em}s_{T}=\sum_{i=1}^{n}s_{i}\\
M & =\sum_{i=1}^{n}s_{i}p_{i}p_{i}^{\top}-s_{T}p_{c}p_{c}^{\top}\\
M\tilde{R} & =\sum_{i=1}^{n}s_{i}(p_{i}-p_{c})y_{i}^{\top}\hat{R}^{\top}\\
\tilde{R}^{\top}\tilde{P}_{\varepsilon} & =\frac{1}{s_{T}}\sum_{i=1}^{n}s_{i}(p_{i}-\hat{R}y_{i}-\hat{P})
\end{cases}\label{eq:NAV_Set_Measurements}
\end{equation}

The vector measurements in \eqref{eq:NAV_Set_Measurements} will be
used to facilitate an explicit observer implementation.
\begin{lem}
	\label{Lemm:SLAM_Lemma1}Consider $\tilde{R}\in\mathbb{SO}\left(3\right)$
	and $M=M^{\top}\in\mathbb{R}^{3\times3}$ as in \eqref{eq:NAV_Set_Measurements}.
	Let $\overline{\mathbf{M}}={\rm Tr}\{M\}\mathbf{I}_{3}-M$ where $\underline{\lambda}_{\overline{\mathbf{M}}}$
	and $\overline{\lambda}_{\overline{\mathbf{M}}}$ are the minimum
	and the maximum eigenvalues of $\overline{\mathbf{M}}$, respectively.
	For $||M\tilde{R}||_{{\rm I}}=\frac{1}{4}{\rm Tr}\{M(\mathbf{I}_{3}-\tilde{R})\}$
	and $\boldsymbol{\Upsilon}(M\tilde{R})=\mathbf{vex}(\boldsymbol{\mathcal{P}}_{a}(M\tilde{R}))$,
	\begin{align}
		\frac{\underline{\lambda}_{\overline{\mathbf{M}}}}{2}(1+{\rm Tr}\{\tilde{R}\})||M\tilde{R}||_{{\rm I}} & \leq||\boldsymbol{\Upsilon}(M\tilde{R})||^{2}\label{eq:SLAM_lemm1_1}\\
		2\overline{\lambda}_{\overline{\mathbf{M}}}||M\tilde{R}||_{{\rm I}} & \geq||\boldsymbol{\Upsilon}(M\tilde{R})||^{2}\label{eq:SLAM_lemm1_2}
	\end{align}
	\begin{proof}See (\cite{hashim2019SO3Wiley}, Lemma 1).\end{proof}
\end{lem}
Let $\lambda(M)=\{\lambda_{1},\lambda_{2},\lambda_{3}\}$ where $\lambda_{3}\geq\lambda_{2}\geq\lambda_{1}$
and $\overline{\mathbf{M}}={\rm Tr}\{M\}\mathbf{I}_{3}-M$. Based
on Assumption \ref{Assum:NAV_1Landmark}, a minimum of two eigenvalues
of $\lambda(M)$ are greater than zero. Hence, (\cite{bullo2004geometric}
page. 553): 1) $\overline{\mathbf{M}}$ is positive-definite, and
2) $\lambda(\overline{\mathbf{M}})=\{\lambda_{3}+\lambda_{2},\lambda_{3}+\lambda_{1},\lambda_{2}+\lambda_{1}\}$
with $\underline{\lambda}_{\overline{\mathbf{M}}}=\lambda_{2}+\lambda_{1}>0$.
\begin{defn}
	\label{def:Unstable-set}\cite{hashim2019SO3Wiley} Define the unstable
	set $\mathbb{U}_{s}\subset\mathbb{SO}\left(3\right)$ as
	\begin{equation}
	\mathbb{U}_{s}=\left\{ \left.\tilde{R}(0)\in\mathbb{SO}\left(3\right)\right|{\rm Tr}\{\tilde{R}(0)\}=-1\right\} \label{eq:Unstable_set_SO3}
	\end{equation}
	where $\tilde{R}(0)\in\mathbb{U}_{s}$ if: $\tilde{R}(0)={\rm diag}(1,-1,-1)$,
	$\tilde{R}(0)={\rm diag}(-1,1,-1)$, or $\tilde{R}(0)={\rm diag}(-1,-1,1)$.
	Based on \eqref{eq:NAV_Ecul_Dist}, $\tilde{R}(0)\in\mathbb{U}_{s}$
	indicates that $||\tilde{R}(0)||_{{\rm I}}=+1$.
\end{defn}

\section{Guaranteed Performance\label{sec:Guaranteed-Performance}}

Define the error vector as follows:
\begin{equation}
e=[e_{1},e_{2},e_{3},e_{4}]^{\top}=\left[||M\tilde{R}||_{{\rm I}},\tilde{P}_{\varepsilon}^{\top}\tilde{R}\right]^{\top}\in\mathbb{R}^{4}\label{eq:NAV_Vec_error}
\end{equation}
This section demonstrates how to guide the error $e$ to reduce systematically
and smoothly from a known large set to a known small set following
the guaranteed measures of transient and steady-state performance.
Hence, define the following prescribed performance function (PPF)
\cite{bechlioulis2008robust}:
\begin{equation}
\xi_{i}(t)=(\xi_{i}^{0}-\xi_{i}^{\infty})\exp(-\ell_{i}t)+\xi_{i}^{\infty},\hspace{1em}\forall i=1,2,\ldots,4\label{eq:NAV_Presc}
\end{equation}
where $\xi_{i}:\mathbb{R}_{+}\rightarrow\mathbb{R}_{+}$, $\xi_{i}(0)=\xi_{i}^{0}\in\mathbb{R}_{+}$
is the upper bound of the known large set, $\xi_{i}^{\infty}\in\mathbb{R}_{+}$
is the upper bound of the small set with $\lim\limits _{t\to\infty}\xi_{i}(t)=\xi_{i}^{\infty}$,
and $\ell_{i}>0$ is the convergence rate of $\xi_{i}\left(t\right)$
from $\xi_{i}^{0}$ to $\xi_{i}^{\infty}$ for all $i=1,2,3,4$. Define
$e_{i}:=e_{i}(t)$, $\xi_{i}:=\xi_{i}(t)$, $\xi=[\xi_{1},\ldots,\xi_{4}]^{\top}$,
$\xi^{0}=[\xi_{1}^{0},\ldots,\xi_{4}^{0}]^{\top}$, $\xi^{\infty}=[\xi_{1}^{\infty},\ldots,\xi_{4}^{\infty}]^{\top}$,
and $\ell=[\ell_{1},\ldots,\ell_{4}]^{\top}$ for all $\xi,\xi^{0},\xi^{\infty},\ell\in\mathbb{R}^{4}$.
$e_{i}$ follows the PPF in \eqref{eq:NAV_Presc}:
\begin{align}
	-\delta\xi_{i}<e_{i}<\xi_{i}, & \text{ if }e_{i}\left(0\right)\geq0\label{eq:NAV_ePos}\\
	-\xi_{i}<e_{i}<\delta\xi_{i}, & \text{ if }e_{i}\left(0\right)<0\label{eq:NAV_eNeg}
\end{align}
where $\delta\in[0,1]$. Figure \ref{fig:NAVPPF_2} illustrates the
concept of PPF for the two scenarios presented in \eqref{eq:NAV_ePos}
and \eqref{eq:NAV_eNeg}.

\begin{figure}[H]
	\centering{}\includegraphics[scale=0.3]{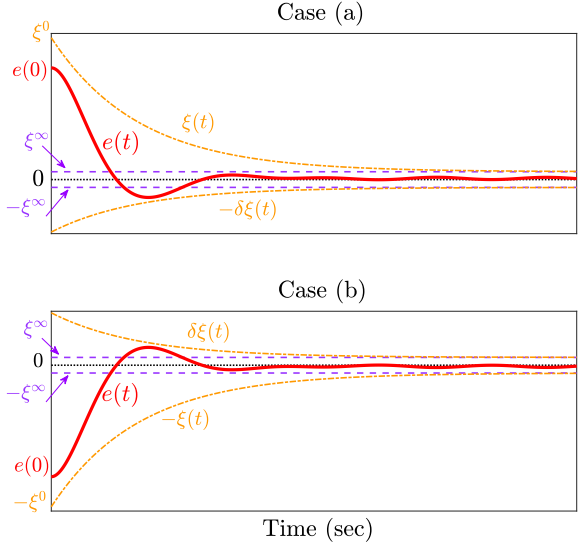} \caption{Guaranteed performance of $e_{i}$ in view of Case (a) Eq. \eqref{eq:NAV_ePos};
		and Case (b) Eq. \eqref{eq:NAV_eNeg}.}
	\label{fig:NAVPPF_2}
\end{figure}

Let $e_{i}$ be defined as 
\begin{equation}
e_{i}=\xi_{i}\mathcal{N}(E_{i})\label{eq:NAV_e_Trans}
\end{equation}
where $\mathcal{N}(E_{i})$ is a smooth function and $E_{i}$ is
a transformed error (unconstrained error). To extract $E_{i}$, $\mathcal{N}(E_{i})$
must follow Assumption \ref{Assum:NAV_2PPF}:

\begin{assum}\label{Assum:NAV_2PPF}\cite{bechlioulis2008robust,hashim2019SE3Det,hashim2020TITS_SLAM} 
	\begin{itemize}
		\item[1)] $\mathcal{N}(E_{i})$ is continuously differentiable and strictly
		increasing. 
		\item[2)] $\mathcal{N}(E_{i})$ is bounded such that 
		\[
		-\delta_{i}<\mathcal{N}(E_{i})<\delta_{i},\,\forall\,e_{i}\left(0\right)\in\mathbb{R}
		\]
		where $\delta_{i}$ is a positive constant. 
		\item[3)] 
		\begin{align*}
			\underset{E_{i}\rightarrow\pm\infty}{\lim}\mathcal{N}(E_{i})=\pm\delta_{i} & ,\,\forall\,e_{i}\left(0\right)\in\mathbb{R}
		\end{align*}
		where $\delta_{i}$ is a positive constant. 
	\end{itemize}
\end{assum}

The following function satisfies Assumption \ref{Assum:NAV_2PPF}:

\begin{equation}
\mathcal{N}(E_{i})=\frac{\delta_{i}\exp(E_{i})-\delta_{i}\exp(-E_{i})}{\exp(E_{i})+\exp(-E_{i})}\label{eq:NAV_Smooth}
\end{equation}
Define $\delta=[\delta_{1},\delta_{2},\delta_{3},\delta_{4}]^{\top}$,
$E=\left[E_{R},E_{P}^{\top}\right]^{\top}$ for all $\delta,E\in\mathbb{R}^{4}$
where $E_{R}=E_{1}\in\mathbb{R}$ and $E_{P}=[E_{2},E_{3},E_{4}]^{\top}\in\mathbb{R}^{3}$.
The inverse transformation of \eqref{eq:NAV_Smooth} is $E_{i}=\mathcal{N}^{-1}(e_{i}/\xi_{i})$
where
\begin{equation}
\begin{aligned}E_{i}= & \frac{1}{2}\text{ln}\frac{\delta_{i}+e_{i}/\xi_{i}}{\delta_{i}-e_{i}/\xi_{i}}\end{aligned}
\label{eq:NAV_trans2}
\end{equation}

\begin{rem}
	\label{rem:NAV_1}\cite{bechlioulis2008robust,hashim2019SE3Det,hashim2020TITS_SLAM}
	Let the unconstrained error be as in \eqref{eq:NAV_trans2}. The transient
	and steady-state performance of $e_{i}$ follow the PPF $\xi_{i}$,
	if and only if $E_{i}\in\mathcal{L}_{\infty}$.
\end{rem}
\begin{rem}
	\label{rem:NAV_2}Note that by definition $e_{1}=||M\tilde{R}||_{{\rm I}}=\frac{1}{4}{\rm Tr}\{M(\mathbf{I}_{3}-\tilde{R})\}\geq0$
	for all $\forall t\geq0$. Also, $e_{i}\leq\xi_{i}$ indicates that
	$(\delta_{i}+e_{i}/\xi_{i})/(\delta_{i}-e_{i}/\xi_{i})\neq1$ for
	all $e_{i}\neq0$, and $(\delta_{i}+e_{i}/\xi_{i})/(\delta_{i}-e_{i}/\xi_{i})=1$
	if $e_{i}=0$. Therefore, $E_{i}\neq0\forall e_{i}\neq0$ and $E_{i}=0$
	if and only if $e_{i}=0$.
\end{rem}
Define {\small{}
	\begin{equation}
	\begin{split}\Delta_{i} & =\frac{1}{2\xi_{i}}\frac{\partial\mathcal{N}^{-1}\left(e_{i}/\xi_{i}\right)}{\partial\left(e_{i}/\xi_{i}\right)}=\frac{1}{2\xi_{i}}(\frac{1}{\delta_{i}+e_{i}/\xi_{i}}+\frac{1}{\delta_{i}-e_{i}/\xi_{i}})\end{split}
	\label{eq:SE3PPF_mu}
	\end{equation}
}Therefore, the dynamics of $\dot{E}_{i}$ can be obtained as
\begin{align}
	\dot{E}_{i} & =\Delta_{i}(\frac{d}{dt}e_{i}-\mu_{i}e_{i})\label{eq:SE3PPF_Trans_dot1}
\end{align}
Consider defining $\mu_{R}=\dot{\xi}_{1}/\xi_{1}$, $\mu_{P}={\rm diag}(\dot{\xi}_{2}/\xi_{2},\dot{\xi}_{3}/\xi_{3},\dot{\xi}_{4}/\xi_{4})$,
$\Delta_{R}=\Delta_{1}$, and $\Delta_{P}={\rm diag}(\Delta_{2},\Delta_{3},\Delta_{4})$
for all $\mu_{R},\Delta_{R}\in\mathbb{R}$ and $\mu_{P},\Delta_{P}\in\mathbb{R}^{3\times3}$.

\section{Nonlinear Stochastic Navigation Filter\label{sec:Non-Nav-Observer1}}

This section pursues three main objectives: 1) providing a comprehensive
discussion of a nonlinear stochastic complementary navigation filter
evolved on $\mathbb{SE}_{2}(3)$, 2) constraining the error vector
defined in \eqref{eq:NAV_Vec_error} to behave following the predefined
transient and steady state measures specified by the user, and 3)
designing a robust filter that is able to produce accurate results
irrespective of the noise level in measurements. The proposed filter
can be utilized for unmanned aerial vehicles and ground robots. Let
$\hat{\sigma}$ denote the estimate of $\sigma$ defined in \eqref{eq:NAV_s}.
Define the error between $\sigma$ and $\hat{\sigma}$ as 
\[
\tilde{\sigma}=\sigma-\hat{\sigma}\in\mathbb{R}^{3}
\]
Based on the set of measurements in \eqref{eq:NAV_Set_Measurements},
the error in \eqref{eq:NAV_Vec_error}, and the transformed error
in \eqref{eq:NAV_trans2}, consider the following nonlinear stochastic
navigation filter evolved directly on the Lie Group of $\mathbb{SE}_{2}\left(3\right)$:

\begin{align}
	\dot{\hat{X}}= & \hat{X}U_{m}-W\hat{X}\label{eq:NAV_Filter1}\\
	= & \left[\begin{array}{ccc}
		\hat{R} & \hat{P} & \hat{V}\\
		0_{1\times3} & 1 & 0\\
		0_{1\times3} & 0 & 1
	\end{array}\right]\left[\begin{array}{ccc}
		[\Omega_{m}\text{\ensuremath{]_{\times}}} & 0_{3\times1} & a_{m}\\
		0_{1\times3} & 0 & 0\\
		0_{1\times3} & 1 & 0
	\end{array}\right]\nonumber \\
	& -\left[\begin{array}{ccc}
		[w_{\Omega}\text{\ensuremath{]_{\times}}} & w_{V} & w_{a}\\
		0_{1\times3} & 0 & 0\\
		0_{1\times3} & 1 & 0
	\end{array}\right]\left[\begin{array}{ccc}
		\hat{R} & \hat{P} & \hat{V}\\
		0_{1\times3} & 1 & 0\\
		0_{1\times3} & 0 & 1
	\end{array}\right]\nonumber 
\end{align}
where $\hat{X}\in\mathbb{SE}_{2}\left(3\right)$, $U_{m}=u([\Omega_{m}\text{\ensuremath{]_{\times}}},0_{3\times1},a_{m},1)\in\mathcal{U}_{m}$
and $W=u([w_{\Omega}\text{\ensuremath{]_{\times}}},w_{V},w_{a},1)\in\mathcal{U}_{m}$.
The equivalent detailed representation of \eqref{eq:NAV_Filter1}
is as follows: 
\begin{equation}
\begin{cases}
\dot{\hat{R}} & =\hat{R}\left[\Omega\right]_{\times}-\left[w_{\Omega}\right]_{\times}\hat{R}\\
\dot{\hat{P}} & =\hat{V}-\left[w_{\Omega}\right]_{\times}\hat{P}-w_{V}\\
\dot{\hat{V}} & =\hat{R}a+\overrightarrow{\mathtt{g}}-\left[w_{\Omega}\right]_{\times}\hat{V}-w_{a}\\
w_{\Omega} & =-k_{w}(\Delta_{R}E_{R}+1)\boldsymbol{\Upsilon}(M\tilde{R})\\
& \hspace{1em}-\frac{\Delta_{R}}{4}\frac{||M\tilde{R}||_{{\rm I}}+2}{||M\tilde{R}||_{{\rm I}}+1}\hat{R}{\rm diag}(\hat{R}^{\top}\boldsymbol{\Upsilon}(M\tilde{R}))\hat{\sigma}\\
w_{V} & =\left[p_{c}\right]_{\times}w_{\Omega}-\frac{k_{v}}{\varepsilon}\Delta_{P}E_{P}-\ell_{P}\tilde{R}^{\top}\tilde{P}_{\varepsilon}\\
w_{a} & =-\overrightarrow{\mathtt{g}}-k_{a}\left(\frac{k_{v}}{\mu}\Delta_{P}+\mathbf{I}_{3}\right)\Delta_{P}E_{P}\\
k_{R} & =\gamma_{\sigma}\frac{||M\tilde{R}||_{{\rm I}}+2}{8}\Delta_{R}^{2}\exp(E_{R})\\
\dot{\hat{\sigma}}_{\Omega} & =k_{R}{\rm diag}(\hat{R}^{\top}\boldsymbol{\Upsilon}(M\tilde{R}))\hat{R}^{\top}\boldsymbol{\Upsilon}(M\tilde{R})-k_{\sigma}\gamma_{\sigma}\hat{\sigma}
\end{cases}\label{eq:NAV_Filter1_Detailed}
\end{equation}
where $k_{w}$, $k_{v}$, $k_{a}$, $\gamma_{\sigma}$, $k_{\sigma}$,
$\mu$, and $\ell_{P}$ are positive constants, $w_{\Omega}\in\mathbb{R}^{3}$,
$w_{V}\in\mathbb{R}^{3}$, and $w_{a}\in\mathbb{R}^{3}$ are the correction
factors $\forall$ $w_{\Omega},w_{V},w_{a}\in\mathbb{R}^{3}$, $\boldsymbol{\Upsilon}(M\tilde{R})=\mathbf{vex}(\boldsymbol{\mathcal{P}}_{a}(M\tilde{R}))$,
$E_{R}=E_{1}$, and $E_{P}=[E_{2},E_{3},E_{4}]^{\top}$. Figure \ref{fig:BD_NAVIGATION}
presents a block diagram of the observer proposed in \eqref{eq:NAV_Filter1}
and detailed in \eqref{eq:NAV_Filter1_Detailed}. 
\begin{figure*}
	\centering{}\includegraphics[scale=0.45]{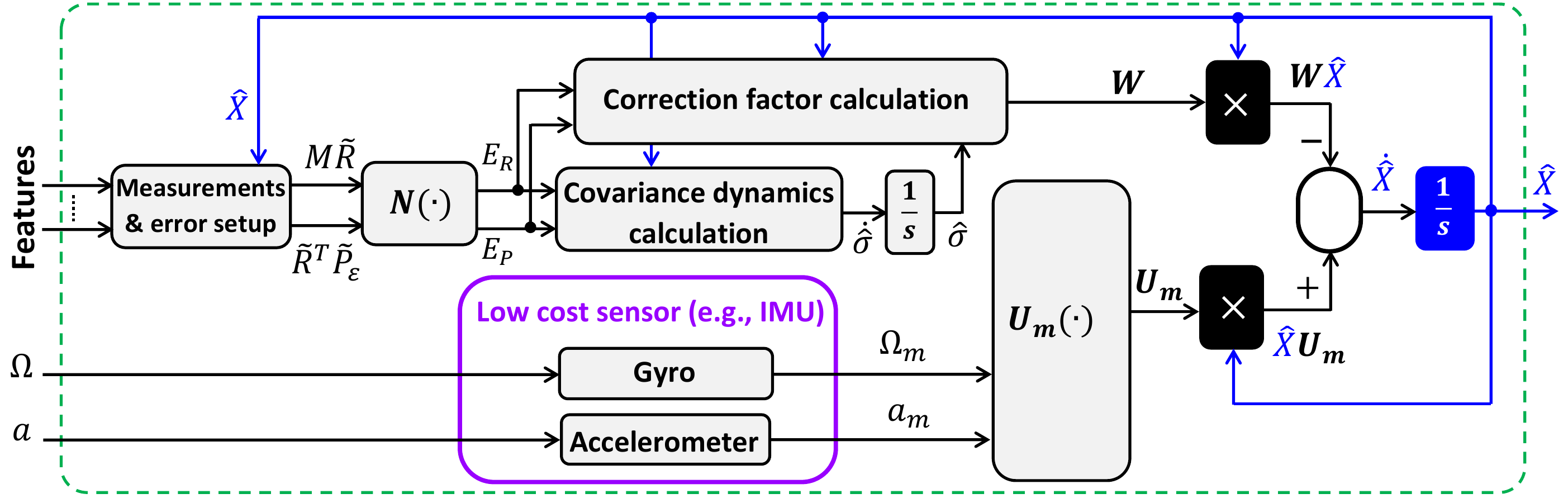}\caption{Block diagram of the proposed nonlinear stochastic observer for inertial
		navigation.}
	\label{fig:BD_NAVIGATION}
\end{figure*}

\begin{thm}
	\label{thm:Theorem1}Consider combining the stochastic dynamics in
	\eqref{eq:NAV_Detailed_True_STCH_dot} with feature measurements ($\overline{y}_{i}=X^{-1}\overline{p}_{i}$)
	for all $i=1,2,\ldots,n$, angular velocity measurements ($\Omega_{m}=\Omega+n_{\Omega}$),
	and acceleration measurements ($a_{m}=a+n_{a}$). Suppose that Assumption
	\ref{Assum:NAV_1Landmark} holds. Consider the nonlinear stochastic
	filter in \eqref{eq:NAV_Filter1} to be geared with $\overline{y}_{i}$,
	$\Omega_{m}$, $a_{m}$, \eqref{eq:NAV_Set_Measurements}, \eqref{eq:NAV_Vec_error},
	and \eqref{eq:NAV_trans2} given that $E_{R},E_{P}\in\mathcal{L}_{\infty}$
	and $\tilde{R}(0)\notin\mathbb{U}_{s}$ as in \eqref{eq:Unstable_set_SO3}.
	Then, 1) the error $e$ in \eqref{eq:NAV_Vec_error} is bounded by
	the dynamically reducing boundaries of $\xi$ in \eqref{eq:NAV_Presc}
	and 2) all the closed loop error signals are almost semi-globally
	uniformly ultimately bounded.
\end{thm}
\begin{proof}Recall $\tilde{R}=R\hat{R}^{\top}$, $\tilde{P}=P-\tilde{R}\hat{P}$,
	$\tilde{P}_{\varepsilon}=\tilde{P}-(\mathbf{I}_{3}-\tilde{R})p_{c}$,
	and $\tilde{V}=V-\tilde{R}\hat{V}$. From \eqref{eq:NAV_Detailed_True_STCH_dot}
	and \eqref{eq:NAV_Filter1_Detailed}, one has
	
	\begin{equation}
	\begin{cases}
	d\tilde{R} & =\tilde{R}[w_{\Omega}]_{\times}dt-\tilde{R}[\hat{R}\mathcal{Q}d\beta_{\Omega}]_{\times}\\
	d\tilde{P} & =(\tilde{V}+\tilde{R}w_{V})dt+\tilde{R}[\hat{P}]_{\times}\hat{R}\mathcal{Q}d\beta_{\Omega}\\
	d\tilde{V} & =((\mathbf{I}_{3}-\tilde{R})g+\tilde{R}w_{a})dt-\tilde{R}\hat{R}\mathcal{Q}d\beta_{a}\\
	& \hspace{1em}-\tilde{R}[\hat{V}]_{\times}\hat{R}\mathcal{Q}d\beta
	\end{cases}\label{eq:NAV_Filter1_Error_dot}
	\end{equation}
	With the objective of designing a filter that directly uses the measurements
	in \eqref{eq:NAV_Set_Measurements}, one finds
	\begin{equation}
	\begin{cases}
	d||M\tilde{R}||_{{\rm I}} & =\frac{1}{2}\boldsymbol{\Upsilon}(M\tilde{R})^{\top}w_{\Omega}dt-\frac{1}{2}\boldsymbol{\Upsilon}(M\tilde{R})^{\top}\hat{R}\mathcal{Q}d\beta_{\Omega}\\
	d\tilde{R}^{\top}\tilde{P}_{\varepsilon} & =\underbrace{(\tilde{R}^{\top}\tilde{V}-[p_{c}-\tilde{R}^{\top}\tilde{P}_{\varepsilon}]_{\times}w_{\Omega}+w_{V})}_{f_{P}}dt\\
	& \hspace{1em}+\underbrace{-[\hat{P}-p_{c}+\tilde{R}^{\top}\tilde{P}_{\varepsilon}]_{\times}\hat{R}}_{h_{P}}\mathcal{Q}d\beta_{\Omega}\\
	d\tilde{R}^{\top}\tilde{V} & =\underbrace{(-[w_{\Omega}]_{\times}\tilde{R}^{\top}\tilde{V}+(\tilde{R}-\mathbf{I}_{3})^{\top}\overrightarrow{\mathtt{g}}+w_{a})}_{f_{V}}dt\\
	& \hspace{1em}+\underbrace{-[\begin{array}{cc}
		[\tilde{R}^{\top}V]_{\times}\hat{R} & \hat{R}\end{array}]}_{h_{V}}[\begin{array}{cc}
	\mathcal{Q}d\beta_{\Omega} & \mathcal{Q}d\beta_{a}\end{array}]^{\top}
	\end{cases}\label{eq:NAV_Filter1_Error_dotf}
	\end{equation}
	Therefore, from \eqref{eq:NAV_Filter1_Error_dotf} and \eqref{eq:SE3PPF_Trans_dot1},
	one obtains the following stochastic dynamics:
	\begin{equation}
	\begin{cases}
	dE_{R} & =\Delta_{R}(d||M\tilde{R}||_{{\rm I}}-\mu_{R}||M\tilde{R}||_{{\rm I}}dt)\\
	dE_{P} & =\underbrace{\Delta_{P}(f_{P}-\mu_{P}\tilde{R}^{\top}\tilde{P}_{\varepsilon})}_{\mathcal{F}_{P}}dt+\Delta_{P}h_{P}\mathcal{Q}d\beta_{\Omega}
	\end{cases}\label{eq:NAV_Filter1_Trans_dot}
	\end{equation}
	Consider the following Lyapunov function candidate $\mathbb{V}=\mathbb{V}(||M\tilde{R}||_{{\rm I}},E_{R},E_{P},\tilde{R}^{\top}\tilde{V},\tilde{\sigma})$:
	\begin{equation}
	\mathbb{V}=\mathbb{V}^{a}+\mathbb{V}^{b}\label{eq:NAV_LyapT_1}
	\end{equation}
	with $\mathbb{V}^{a}=\mathbb{V}^{a}(||M\tilde{R}||_{{\rm I}},E_{R},\tilde{\sigma})$
	and $\mathbb{V}^{b}=\mathbb{V}^{b}(E_{P},\tilde{R}^{\top}\tilde{V})$.
	The first part of \eqref{eq:NAV_LyapT_1} has the map $\mathbb{V}^{a}:\mathbb{SO}\left(3\right)\times\mathbb{R}\times\mathbb{R}^{3}\rightarrow\mathbb{R}_{+}$
	and is selected as
	\begin{equation}
	\mathbb{V}^{a}=\exp(E_{R})||M\tilde{R}||_{{\rm I}}+\frac{1}{2\gamma_{\sigma}}||\tilde{\sigma}||^{2}\label{eq:NAV_LyapR_1}
	\end{equation}
	where $E_{R}\geq0$ and $||M\tilde{R}||_{{\rm I}}\geq0\forall t\geq0$,
	$E_{R}\neq0$ for all $||M\tilde{R}||_{{\rm I}}\neq0$, and $E_{R}=0$
	only at $||M\tilde{R}||_{{\rm I}}=0$ as discussed in Section \ref{sec:Guaranteed-Performance}.
	The first and the second partial derivatives of $\exp(E_{R})||M\tilde{R}||_{{\rm I}}$
	are $(||M\tilde{R}||_{{\rm I}}+1)\exp(E_{R})$ and $(||M\tilde{R}||_{{\rm I}}+2)\exp(E_{R})$,
	respectively. In view of \eqref{eq:NAV_dVfunction_Lyap}, \eqref{eq:NAV_Filter1_Error_dotf},
	and \eqref{eq:NAV_Filter1_Trans_dot}, one has{\small{}
		\begin{align}
			\mathcal{L}\mathbb{V}^{a}= & \exp(E_{R})(||M\tilde{R}||_{{\rm I}}+1)\Delta_{R}(\frac{1}{2}\boldsymbol{\Upsilon}(M\tilde{R})^{\top}w_{\Omega}-\mu_{R}||M\tilde{R}||_{{\rm I}})\nonumber \\
			& +k_{R}\boldsymbol{\Upsilon}(M\tilde{R})^{\top}\hat{R}\mathcal{Q}^{2}\hat{R}^{\top}\boldsymbol{\Upsilon}(M\tilde{R})-\frac{1}{\gamma_{\sigma}}\tilde{\sigma}^{\top}\dot{\hat{\sigma}}\nonumber \\
			\leq & \exp(E_{R})(||M\tilde{R}||_{{\rm I}}+1)\Delta_{R}(\frac{1}{2}\boldsymbol{\Upsilon}(M\tilde{R})^{\top}w_{\Omega}-\mu_{R}||M\tilde{R}||_{{\rm I}})\nonumber \\
			& +k_{R}\boldsymbol{\Upsilon}(M\tilde{R})^{\top}\hat{R}{\rm diag}(\sigma)\hat{R}^{\top}\boldsymbol{\Upsilon}(M\tilde{R})-\frac{1}{\gamma_{\sigma}}\tilde{\sigma}^{\top}\dot{\hat{\sigma}}\label{eq:NAV_LyapR_1dot_1}
		\end{align}
	}where $\mathcal{Q}^{2}$ has been replaced by ${\rm diag}(\sigma)$
	in \eqref{eq:NAV_s}. Replacing $w_{\Omega}$ and $\dot{\hat{\sigma}}$
	with their definitions in \eqref{eq:NAV_Filter1_Detailed} and recalling
	\eqref{eq:SLAM_lemm1_1} in Lemma \ref{Lemm:SLAM_Lemma1}, one has
	\begin{align}
		\mathcal{L}\mathbb{V}^{a}\leq & -(1+{\rm Tr}\{\tilde{R}\})\frac{(k_{w}-c_{m})\underline{\lambda}_{\overline{\mathbf{M}}}\Delta_{R}^{2}}{4}\exp(E_{R})||M\tilde{R}||_{{\rm I}}\nonumber \\
		& +k_{\sigma}\tilde{\sigma}^{\top}\hat{\sigma}\label{eq:NAV_LyapR_1dot_2}
	\end{align}
	Since $|\mu_{R}|\leq\ell_{1}$, let $c_{m}=\frac{4\ell_{1}}{\underline{\lambda}_{\overline{\mathbf{M}}}(1+{\rm Tr}\{\tilde{R}\})}$
	and select $k_{w}>c_{m}$. Also, in view of Young's inequality, one
	has $k_{\sigma}\tilde{\sigma}^{\top}\sigma\leq(k_{\sigma}/2)||\sigma||^{2}+(k_{\sigma}/2)||\tilde{\sigma}||^{2}$.
	Define $c_{R}=\underline{\lambda}_{\overline{\mathbf{M}}}(1+{\rm Tr}\{\tilde{R}\})$.
	\begin{align}
		\mathcal{L}\mathbb{V}^{a}\leq & -\frac{(k_{w}-c_{m})c_{R}\Delta_{R}^{2}}{4}\exp(E_{R})||M\tilde{R}||_{{\rm I}}-\frac{k_{\sigma}}{2}||\tilde{\sigma}||^{2}\nonumber \\
		& +(k_{\sigma}/2)||\sigma||^{2}\label{eq:NAV_LyapR_1dot_Final}
	\end{align}
	Turning our attention to the second part of \eqref{eq:NAV_LyapT_1},
	define the following Lyapunov function candidate:
	\begin{equation}
	\mathbb{V}^{b}=\frac{1}{4}||E_{P}||^{4}+\frac{1}{4k_{a}}||\tilde{R}^{\top}\tilde{V}||^{4}-\frac{1}{\mu}||\tilde{R}^{\top}\tilde{V}||^{2}\tilde{V}^{\top}\tilde{R}E_{P}\label{eq:NAV_LyapPV_1}
	\end{equation}
	In view of \eqref{eq:NAV_dVfunction_Lyap}, \eqref{eq:NAV_Filter1_Error_dotf},
	and \eqref{eq:NAV_Filter1_Trans_dot}, one has
	\begin{align}
		\mathcal{L}\mathbb{V}^{b} & =(||E_{P}||^{2}E_{P}^{\top}-\frac{1}{\mu}||\tilde{R}^{\top}\tilde{V}||^{2}\tilde{V}^{\top}\tilde{R})\mathcal{F}_{P}\nonumber \\
		& +(\frac{1}{k_{a}}||\tilde{R}^{\top}\tilde{V}||^{2}\tilde{V}^{\top}\tilde{R}-\frac{1}{\mu}||\tilde{R}^{\top}\tilde{V}||^{2}E_{P}^{\top})f_{V}\nonumber \\
		& +\frac{1}{2}{\rm Tr}\left\{ (||E_{P}||^{2}\mathbf{I}_{3}+2E_{P}E_{P}^{\top})\Delta_{P}h_{P}\mathcal{Q}^{2}h_{P}^{\top}\Delta_{P}\right\} \nonumber \\
		& +\frac{1}{2k_{a}}{\rm Tr}\left\{ (||\tilde{R}^{\top}\tilde{V}||^{2}\mathbf{I}_{3}+2E_{P}\tilde{V}^{\top}\tilde{R})h_{V}\mathcal{Q}^{2}h_{V}^{\top}\right\} \label{eq:NAV_LyapPV_1dot_1}
	\end{align}
	The equation in \eqref{eq:NAV_LyapPV_1dot_1} can be presented in
	an inequality form as below:
	\begin{align}
		\mathcal{L}\mathbb{V}^{b}\leq & -\underline{c}_{2}(\frac{4k_{v}-3}{4})||E_{P}||^{4}\nonumber \\
		& -(\frac{\underline{c}_{1}}{\mu}-\frac{\overline{c}_{2}}{2\mu k_{v}}-\frac{||g||}{2k_{a}}-\frac{\underline{c}_{1}}{4k_{a}})||\tilde{R}^{\top}\tilde{V}||^{4}\nonumber \\
		& +(\frac{\overline{c}_{1}+\overline{c}_{2}}{\mu}+\frac{\underline{c}_{1}}{4k_{a}}+\frac{||g||\overline{\lambda}_{M}}{2\mu})||\tilde{R}^{\top}\tilde{V}||^{2}||E_{P}||^{2}\nonumber \\
		& +(\frac{||g||\overline{\lambda}_{M}}{2k_{a}}+\frac{||g||\overline{\lambda}_{M}}{2\mu})||\tilde{R}^{\top}\tilde{V}||^{2}||\mathbf{I}_{3}-\tilde{R}||_{F}\nonumber \\
		& +(\frac{3c_{P}}{2\underline{c}_{2}}+\frac{c_{V}}{4k_{a}\underline{c}_{1}})||\sigma||^{2}\label{eq:NAV_LyapPV_1dot_1-1}
	\end{align}
	where $\underline{c}_{1}=\min\{\Delta_{P}\}$, $\overline{c}_{1}=\max\{\Delta_{P}\}$,
	$\underline{c}_{2}=\min\{\Delta_{P}^{2}\}$, $\overline{c}_{2}=\max\{\Delta_{P}^{2}\}$,
	$c_{V}=\sup_{t\geq0}(1+||V||^{2})$, and $c_{P}=\sup_{t\geq0}||P-p_{c}||^{2}$.
	Let $c_{x}=\max\{\underline{c}_{1},\overline{c}_{1},\overline{c}_{2},||g||\overline{\lambda}_{M}\}$
	and $c_{d}=\max\{2\frac{||g||\overline{\lambda}_{M}}{\mu},2\frac{||g||\overline{\lambda}_{M}}{k_{a}}\}$.
	Consequently, one obtains
	
	\begin{align}
		\mathcal{L}\mathbb{V}^{b}\leq & -\underline{c}_{2}(k_{v}-\frac{3}{4})||E_{P}||^{4}-(\frac{\underline{c}_{1}}{\mu}-\frac{c_{x}}{k_{a}})||\tilde{R}^{\top}\tilde{V}||^{4}\nonumber \\
		& +\frac{5c_{x}}{4\mu}||\tilde{R}^{\top}\tilde{V}||^{2}||E_{P}||^{2}+2c_{d}||\tilde{R}^{\top}\tilde{V}||^{2}||\mathbf{I}_{3}-\tilde{R}||_{F}\nonumber \\
		& +(\frac{3c_{P}}{2\underline{c}_{2}}+\frac{c_{V}}{4k_{a}\underline{c}_{1}})||\sigma||^{2}\label{eq:NAV_LyapPV_1dot_Final}
	\end{align}
	such that
	\begin{align}
		\mathcal{L}\mathbb{V}^{b}\leq & 2c_{d}||\tilde{R}^{\top}\tilde{V}||^{2}||\mathbf{I}_{3}-\tilde{R}||_{F}+(\frac{3c_{P}}{2\underline{c}_{2}}+\frac{c_{V}}{4k_{a}\underline{c}_{1}})||\sigma||^{2}\nonumber \\
		& -\varepsilon_{PV}^{\top}\underbrace{\left[\begin{array}{cc}
				\underline{c}_{2}\frac{4k_{v}-3}{4} & \frac{5c_{x}}{8\mu}\\
				\frac{5c_{x}}{8\mu} & \frac{\underline{c}_{1}}{\mu}-\frac{c_{x}}{k_{v}}
			\end{array}\right]}_{A_{1}}\varepsilon_{PV}\label{eq:NAV_LyapPV_1dot_Final-1}
	\end{align}
	where $\varepsilon_{PV}=[||E_{P}||,||\tilde{R}^{\top}\tilde{V}||]^{\top}$.
	It can be deduced that $A_{1}$ can be made positive by selecting
	$k_{v}>\max\{3/4,\mu c_{x}/\underline{c}_{1},c_{x}\mu^{2}/(\mu\underline{c}_{1}-c_{x})\}$.
	Recalling \eqref{eq:NAV_LyapT_1}, \eqref{eq:NAV_LyapR_1dot_Final},
	and \eqref{eq:NAV_LyapPV_1dot_Final}, one obtains
	\begin{align}
		\mathcal{L}\mathbb{V}\leq & -\frac{(k_{w}-c_{m})c_{R}\Delta_{R}^{2}}{4}\exp(E_{R})||M\tilde{R}||_{{\rm I}}-\underline{\lambda}_{A_{1}}||\varepsilon_{PV}||^{2}\nonumber \\
		& +2c_{d}||\tilde{R}^{\top}\tilde{V}||^{2}||\mathbf{I}_{3}-\tilde{R}||_{F}-\frac{k_{\sigma}}{2}||\tilde{\sigma}||^{2}\nonumber \\
		& +(\frac{k_{\sigma}}{2}+\frac{3c_{P}}{2\underline{c}_{2}}+\frac{c_{V}}{4k_{a}\underline{c}_{1}})||\sigma||^{2}\label{eq:NAV_LyapPV_1dot_2}
	\end{align}
	which is 
	\begin{align}
		\mathcal{L}\mathbb{V}\leq & -\varepsilon_{T}^{\top}\underbrace{\left[\begin{array}{cc}
				\frac{(k_{w}-c_{m})c_{R}\underline{c}_{2}}{4} & c_{d}\mathbf{1}_{1\times2}\\
				c_{d}\mathbf{1}_{2\times1} & \underline{\lambda}_{1}\mathbf{I}_{2}
			\end{array}\right]}_{A_{2}}\varepsilon_{T}-\frac{k_{\sigma}}{2}||\tilde{\sigma}||^{2}\nonumber \\
		& +(\frac{3c_{P}}{2\underline{c}_{2}}+\frac{c_{V}}{4k_{a}\underline{c}_{1}}+\frac{k_{\sigma}}{2})||\sigma||^{2}\nonumber \\
		\leq & -\underline{\lambda}_{A_{2}}||\varepsilon_{T}||^{2}-\frac{k_{\sigma}}{2}||\tilde{\sigma}||^{2}+\eta_{2}||\sigma||^{2}\label{eq:NAV_LyapT_1dot_Final}
	\end{align}

\begin{algorithm}
	\caption{\label{alg:Alg_Disc0}Discrete nonlinear stochastic filter with guaranteed
		performance}
	{\footnotesize
		
		\textbf{Initialization}:
		\begin{enumerate}
			\item[{\footnotesize{}1:}] Set $\hat{R}[0]\in\mathbb{SO}\left(3\right)$, and $\hat{P}[0],\hat{V}[0],\hat{\sigma}[0]\in\mathbb{R}^{3}$
			\item[{\footnotesize{}2:}] Start with $k=0$ and select the design parameters
		\end{enumerate}
		\textbf{while }(1)\textbf{ do}
		\begin{enumerate}
			\item[] \textcolor{blue}{/{*} Prediction step {*}/}
			\item[{\footnotesize{}3:}] $\hat{X}_{k|k}=\left[\begin{array}{ccc}
			\hat{R}_{k|k} & \hat{P}_{k|k} & \hat{V}_{k|k}\\
			0_{1\times3} & 1 & 0\\
			0_{1\times3} & 0 & 1
			\end{array}\right]$ and \\
			$\hat{U}_{k}=u([\Omega_{m}[k]\text{\ensuremath{]_{\times}}},0_{3\times1},a_{m}[k],1)$,
			see \eqref{eq:NAV_u}
			\item[{\footnotesize{}4:}] $\hat{X}_{k+1|k}=\hat{X}_{k|k}\exp(\hat{U}_{k}\Delta t)$
			\item[] \textcolor{blue}{/{*} Update step {*}/}
			\item[{\footnotesize{}5:}] $\begin{cases}
			p_{c} & =\frac{1}{s_{T}}\sum_{i=1}^{n}s_{i}p_{i}[k],\hspace{1em}s_{T}=\sum_{i=1}^{n}s_{i}\\
			M_{k} & =\sum_{i=1}^{n}s_{i}p_{i}[k]p_{i}^{\top}[k]-s_{T}p_{c}p_{c}^{\top}\\
			M\tilde{R}_{k} & =\sum_{i=1}^{n}s_{i}\left(p_{i}[k]-p_{c}\right)y_{i}^{\top}[k]\hat{R}_{k+1|k}^{\top}\\
			\tilde{R}^{\top}\tilde{P}_{\varepsilon}[k] & =\frac{1}{s_{T}}\sum_{i=1}^{n}s_{i}(p_{i}[k]-\hat{R}_{k+1|k}y_{i}[k]-\hat{P}_{k+1|k})
			\end{cases}$
			\item[{\footnotesize{}6:}] $[e_{1}[k],e_{2}[k],e_{3}[k],e_{4}[k]]^{\top}=\left[||M\tilde{R}_{k}||_{{\rm I}},\tilde{P}_{\varepsilon}^{\top}\tilde{R}[k]\right]^{\top}$
			\item[{\footnotesize{}7:}] \textbf{for} $i=1:4$ \hspace{0.5cm}\textcolor{blue}{/{*} Guaranteed
				Performance{*}/}
			\item[{\footnotesize{}8:}] \hspace{0.5cm}$\xi_{i}[k]=(\xi_{0}-\xi_{\infty})\exp(-\ell k\Delta t)+\xi_{\infty}$
			\item[{\footnotesize{}9:}] \hspace{0.5cm}\textbf{if} $e_{i}[k]>\xi_{i}[k]$ \textbf{then}
			\item[{\footnotesize{}10:}] \hspace{0.5cm}\hspace{0.5cm}$\xi_{i}[k]=e_{i}[k]+\epsilon$,\hspace{0.5cm}\textcolor{blue}{{}
				/{*} $\epsilon$ is a small constant {*}/}
			\item[{\footnotesize{}11:}] \hspace{0.5cm}\textbf{end if}
			\item[{\footnotesize{}12:}] \hspace{0.5cm}$\begin{cases}
			E_{i} & =\frac{1}{2}\text{ln}\frac{\delta_{i}+e_{i}[k]/\xi_{i}[k]}{\delta_{i}-e_{i}[k]/\xi_{i}[k]}\\
			\Delta_{i} & =\frac{1}{2\xi_{i}[k]}(\frac{1}{\delta_{i}+e_{i}[k]/\xi_{i}[k]}+\frac{1}{\delta_{i}-e_{i}[k]/\xi_{i}[k]})
			\end{cases}$
			\item[{\footnotesize{}13:}] \textbf{end for}
			\item[{\footnotesize{}14:}] Set $E_{R}[k]=E_{1}$, $E_{P}[k]=[E_{2},E_{3},E_{4}]^{\top}$, $\Delta_{R}[k]=\Delta_{1}$,
			and $\Delta_{P}[k]={\rm diag}(\Delta_{2},\Delta_{3},\Delta_{4})$.
			\item[{\footnotesize{}15:}] $\begin{cases}
			w_{\Omega}[k] & =-k_{w}(E_{R}[k]+1)\Delta_{R}[k]\boldsymbol{\Upsilon}(M\tilde{R}_{k})\\
			& \hspace{1em}-\frac{\Delta_{R}[k]}{4}\frac{||M\tilde{R}_{k}||_{{\rm I}}+2}{||M\tilde{R}_{k}||_{{\rm I}}+1}\hat{R}_{k}{\rm diag}(\hat{R}_{k}^{\top}\boldsymbol{\Upsilon}(M\tilde{R}_{k}))\hat{\sigma}_{k}\\
			w_{V}[k] & =[p_{c}[k]]_{\times}w_{\Omega}[k]-\ell_{P}\tilde{R}^{\top}\tilde{P}_{\varepsilon}[k]\\
			& \hspace{1em}-\frac{k_{v}}{\varepsilon}\Delta_{P}[k]E_{P}[k]\\
			w_{a}[k] & =-k_{a}(\frac{k_{v}}{\mu}\Delta_{P}[k]+\mathbf{I}_{3})\Delta_{P}[k]E_{P}[k]
			\end{cases}$
			\item[{\footnotesize{}16:}] $W_{k}=u([w_{\Omega}[k]\text{\ensuremath{]_{\times}}},w_{V}[k],w_{a}[k]-\overrightarrow{\mathtt{g}},0)$,
			see \eqref{eq:NAV_u}
			\item[{\footnotesize{}17:}] $k_{R}=\gamma_{\sigma}\frac{||M\tilde{R}_{k}||_{{\rm I}}+2}{8}\Delta_{R}[k]^{2}\exp(E_{R}[k])$
			\item[{\footnotesize{}18:}] {\small{}$\hat{\sigma}_{k+1}=\hat{\sigma}_{k}+\Delta tk_{R}{\rm diag}(\hat{R}_{k+1|k}^{\top}\boldsymbol{\Upsilon}(M\tilde{R}_{k}))\hat{R}_{k+1|k}^{\top}\boldsymbol{\Upsilon}(M\tilde{R}_{k})-\Delta tk_{\sigma}\gamma_{\sigma}\hat{\sigma}_{k}$}{\small\par}
			\item[{\footnotesize{}19:}] $\hat{X}_{k+1|k+1}=\exp(-W_{k}\Delta t)\hat{X}_{k+1|k}$
			\item[{\footnotesize{}20:}] $k=k+1$
		\end{enumerate}
		\textbf{end while}}
\end{algorithm}

	where $\varepsilon_{T}=[\sqrt{\exp(E_{\tilde{R}})||M\tilde{R}||_{{\rm I}}},||E_{P}||^{2},||\tilde{R}^{\top}\tilde{V}||^{2}]^{\top}$
	and $\eta_{2}=\frac{3c_{P}}{2\underline{c}_{2}}+\frac{c_{V}}{4k_{a}\underline{c}_{1}}+\frac{k_{\sigma}}{2}$.
	$A_{2}$ is made positive by selecting $(k_{w}-c_{m})c_{R}\underline{c}_{2}\underline{\lambda}_{1}>c_{d}^{2}$.
	Based on \eqref{eq:NAV_LyapT_1dot_Final}, one finds
	\begin{align}
		\mathcal{L}\mathbb{V}\leq & -\underline{\lambda}_{A_{2}}(\exp(E_{\tilde{R}})||M\tilde{R}||_{{\rm I}}+||E_{P}||^{4}+||\tilde{R}^{\top}\tilde{V}||^{4})\nonumber \\
		& -\frac{k_{\sigma}}{2}||\tilde{\sigma}||^{2}+\eta_{2}||\sigma||^{2}\label{eq:NAV_LyapT_1dot_Final-1}
	\end{align}
	which means that
	\begin{equation}
	\frac{d\mathbb{E}[\mathbb{V}]}{dt}=\mathbb{E}[\mathcal{L}\mathbb{V}]\leq-\eta_{1}\mathbb{E}[\mathbb{V}]+\eta_{2}\label{eq:SLAM_Lyap2_final3}
	\end{equation}
	with $\eta_{1}=\min\{\underline{\lambda}_{A_{2}},k_{\sigma}/2\}$.
	In view of Lemma \ref{Lemm:NAV_deng}, the inequality in \eqref{eq:SLAM_Lyap2_final3}
	shows that 
	\[
	0\leq\mathbb{E}[\mathbb{V}(t)]\leq\mathbb{V}(0){\rm exp}(-\eta_{1}t)+\frac{\eta_{2}}{\eta_{1}},\,\forall t\geq0.
	\]
	Therefore, the vector $[\varepsilon_{T}^{\top},||\tilde{\sigma}||]$
	is almost semi-globally uniformly ultimately bounded which completes
	the proof.\end{proof}

The filter in \eqref{eq:NAV_Filter1_Detailed} is proposed in a continuous
form. For the purposes of low-cost electronic kit implementation,
a discrete version practicable at a low sampling rate is necessary.
Let $\Delta t$ denote a small sample time. The complete discrete
implementation steps are detailed in Algorithm \ref{alg:Alg_Disc0}.

\section{Experimental Results \label{sec:SE3_Simulations}}

This section experimentally demonstrates the performance of the proposed
nonlinear stochastic filter for inertial navigation with guaranteed
performance on the Lie group of $\mathbb{SE}_{2}\left(3\right)$.
The proposed stochastic filter in its discrete form (Algorithm \ref{alg:Alg_Disc0})
has been tested on the real-word data obtained from the EuRoC dataset
\cite{Burri2016Euroc}. The dataset is comprised of the ground truth,
IMU measurements provided by ADIS16448 at a sampling rate of 200 Hz,
and stereo images collected by MT9V034 at a sampling rate of 20 Hz.
The camera parameters were calibrated using the Stereo Camera Calibrator
Application ensuring that the images are not distorted. Figure \ref{fig:NAV_Camera}
presents right and left feature detection. The detection of each photo
in the EuRoC dataset \cite{Burri2016Euroc} was tracked based on minimum
eigenvalue feature detection using Kanade-Lucas-Tomasi (KLT) \cite{shi1994good}.
The noise inherent to low-cost IMUs has been supplemented with additional
noise of $n_{\Omega}=\mathcal{N}(0,0.11)$ (rad/sec) for the measurements
of angular velocity while $n_{a}=\mathcal{N}(0,0.1)$ (m/sec$^{2}$)
for the measurements of acceleration. Note that $\mathcal{N}(0,0.11)$
is a notation of normally distributed noise with a zero mean and a
standard deviation of $0.11$. Owing to the lack of real-world features
in the dataset, feature locations are selected arbitrarily. The maximum
number of features were set to 30. The downloadable datasets can be accessed at the following \href{https://projects.asl.ethz.ch/datasets/doku.php?id=kmavvisualinertialdatasets}{URL}. The initial covariance estimate
is set to $\hat{\sigma}\left(0\right)=[0,0,0]^{\top}$, and the design
parameters are selected as $k_{w}=3$, $k_{v}=3$, $k_{a}=20$, $\gamma_{\sigma}=3$,
$k_{\sigma}=0.1$, $\mu=0.8$, $\varepsilon=0.8$, $\ell=[1,1,1,1]^{\top}$,
$\xi^{\infty}=[0.03,0.1,0.1,0.1]^{\top}$, and 
\[
\xi_{i}^{0}=\delta=\left[\begin{array}{c}
1.2||M\tilde{R}(0)||_{{\rm I}}+0.5\\
2\tilde{R}(0)^{\top}\tilde{P}_{\varepsilon}(0)+2\times1_{3\times1}
\end{array}\right]
\]
For large error in initialization, set the initial estimates of attitude,
position, and linear velocity to
\begin{align*}
	\hat{R}_{0}=\hat{R}\left(0\right) & =\mathbf{I}_{3},\hspace{1em}\hat{P}\left(0\right)=\hat{V}\left(0\right)=[0,0,0]^{\top}
\end{align*}

\begin{figure}[H]
	\centering{}\includegraphics[scale=0.2]{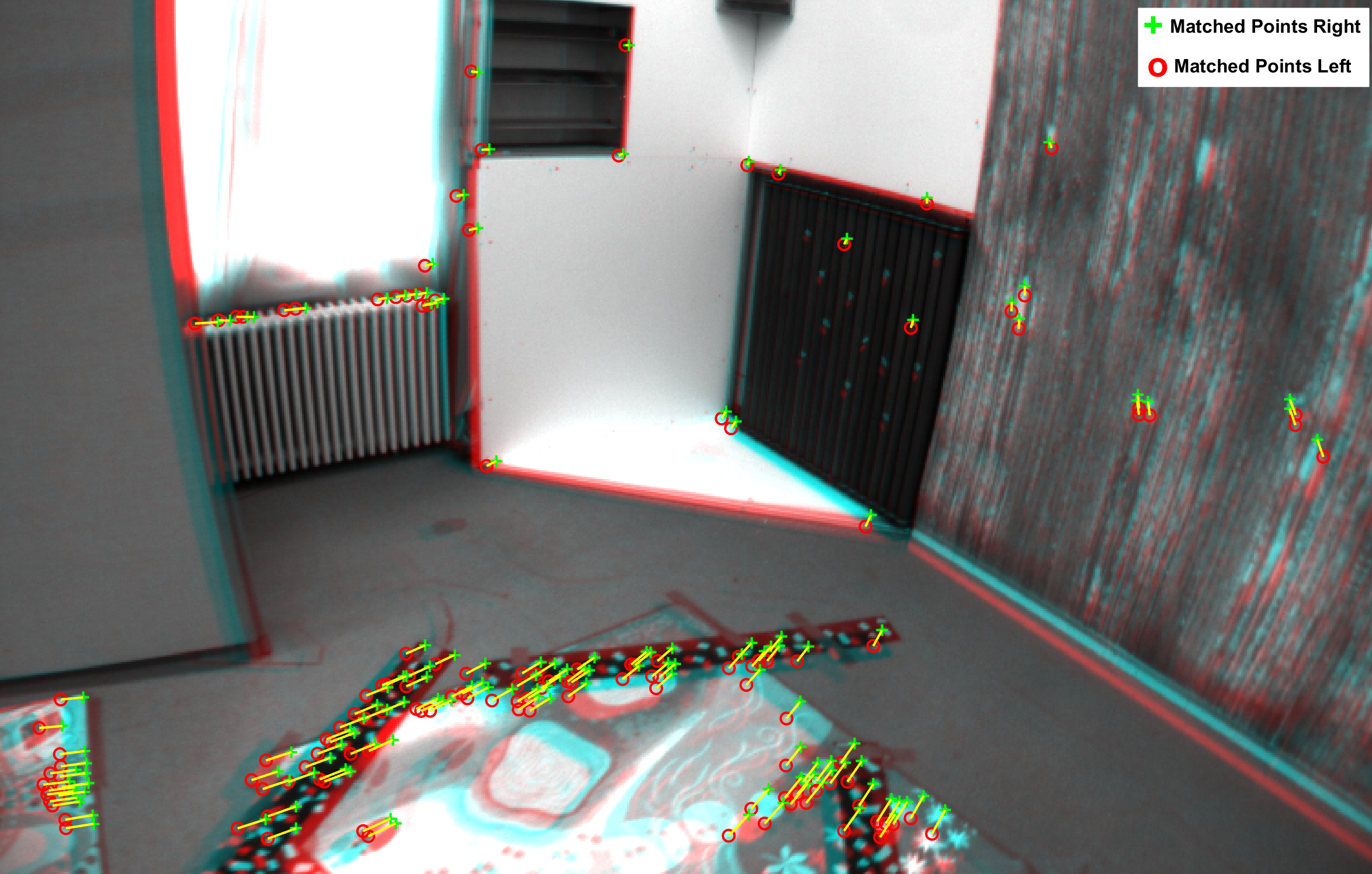}\caption{Illustration of right and left feature detection. The detection tracking
		is performed using a stereo camera through the Computer Vision System
		Toolbox with MATLAB R2020a. The photo is part of the EuRoC dataset
		\cite{Burri2016Euroc}.}
	\label{fig:NAV_Camera}
\end{figure}

\begin{figure*}
	\centering{}\includegraphics[scale=0.30]{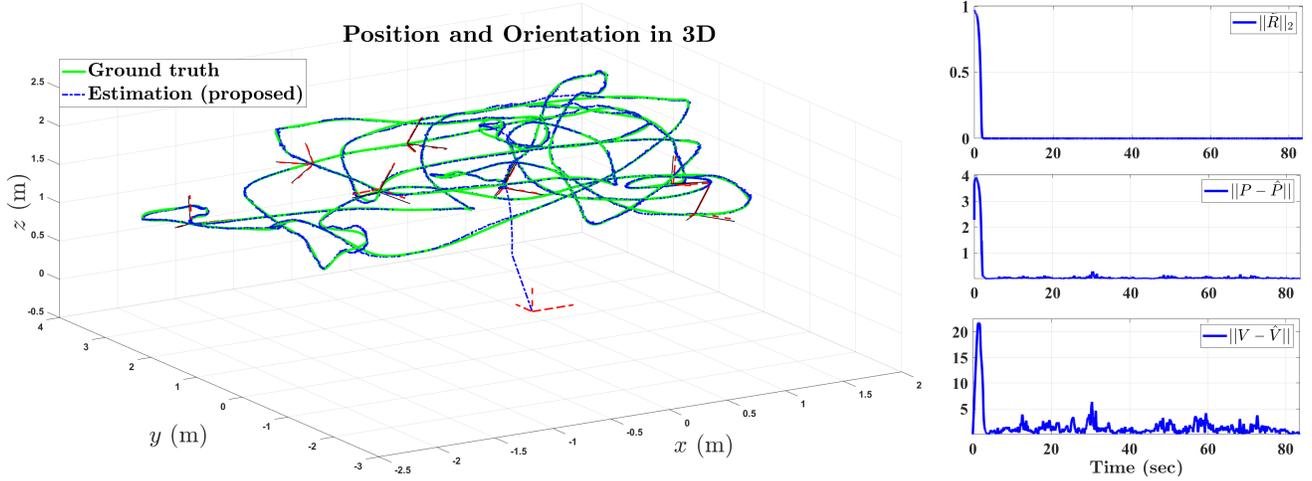}\caption{Experimental validation using Vicon Room (V1\_02\_medium) dataset.
		In the left: the true vehicle trajectory (green solid-line) is plotted
		against the trajectory estimated by the proposed nonlinear stochastic
		discrete navigation filter (Algorithm \ref{alg:Alg_Disc0}; blue dashed-line).
		The true and the estimated final destinations are marked with a green
		circle and a blue star $\star$, respectively. Three axes illustrate
		the vehicle's true (black solid-line) and estimated (red dashed-line)
		attitude. On the right: the error components of attitude $||R\hat{R}^{\top}||_{{\rm I}}$,
		position $||P-\hat{P}||$, and velocity $||V-\hat{V}||$ are plotted
		in blue solid-line. To access experiment demonstration
		visit the \href{https://youtu.be/ISUsnQbvz74}{Video Link}.}
	\label{fig:NAV_3D}
\end{figure*}

The left portion of Figure \ref{fig:NAV_3D} illustrates rapid adaptation
from large initial error in attitude and position and strong tracking
capabilities of the filter to the true trajectory. The right portion
of Figure \ref{fig:NAV_3D} reveals impressive convergence of the
error components $||R\hat{R}^{\top}||_{{\rm I}}$, $||P-\hat{P}||$,
and $||V-\hat{V}||$ from large initial error to the close neighborhood
of the origin. Also, Figure \ref{fig:NAV_3D} confirms the robustness
of the filter against the high levels of measurement noise. Figure
\ref{fig:NAV_Eul} shows remarkable tracking performance of the filter
estimated orientation in terms of Euler angles ($\hat{\phi}$, $\hat{\theta}$,
and $\hat{\psi}$) in comparison with the true vehicle's orientation
($\phi$, $\theta$, and $\psi$). Likewise, Figure \ref{fig:NAV_Pos}
illustrates impressive tracking performance of the filter estimated
position ($\hat{x}$, $\hat{y}$, and $\hat{z}$) against the true
vehicle's position ($x$, $y$, and $z$). To access experiment demonstration
visit the \href{https://youtu.be/ISUsnQbvz74}{Video Link}. The real-world
dataset experiment demonstrates the ability of the proposed stochastic
filter to produce good results at low sampling rates and in presence
of large initial error and measurement uncertainties. The performed
testing indicates that the proposed filter is computationally cheap.
As such, the proposed solution can be implemented using an inexpensive
kit.

\begin{figure}[h]
	\centering{}\includegraphics[scale=0.32]{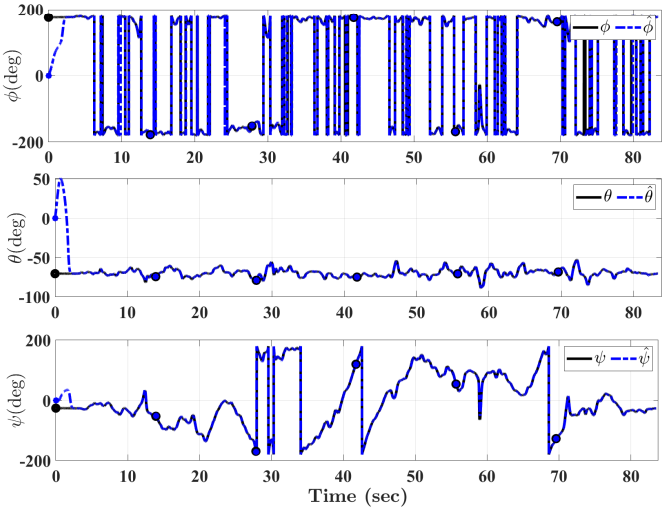}\caption{Experimental validation using Vicon Room (V1\_02\_medium) dataset.
		Euler angles: True trajectory shown as a black solid line ($\phi$,
		$\theta$, and $\psi$) vs estimated trajectory shown as a blue dashed-line
		($\hat{\phi}$, $\hat{\theta}$, and $\hat{\psi}$).}
	\label{fig:NAV_Eul}
\end{figure}

\begin{figure}[h]
	\centering{}\includegraphics[scale=0.32]{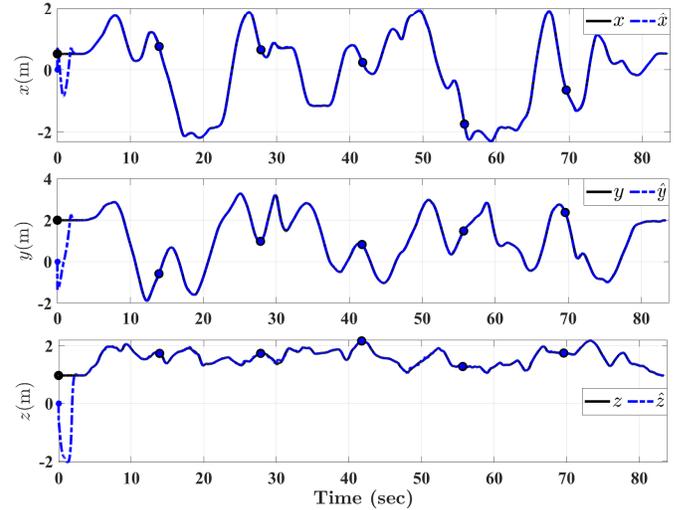}\caption{Experimental validation using Vicon Room (V1\_02\_medium) dataset.
		Position in 3D space: True trajectory shown as a black solid line
		($x$, $y$, and $z$) vs estimated trajectory shown as a blue dashed-line
		($\hat{x}$, $\hat{y}$, and $\hat{z}$).}
	\label{fig:NAV_Pos}
\end{figure}

\section{Conclusion \label{sec:SE3_Conclusion}}

This paper addressed the problem of attitude, position, and linear
velocity estimation of a rigid-body navigating with six degrees of
freedom (6 DoF). A geometric nonlinear stochastic navigation filter
on the matrix Lie group of $\mathbb{SE}_{2}(3)$ with guaranteed transient
and steady-state performance has been proposed. The closed loop error
signals have been shown to be almost semi-globally uniformly ultimately
bounded in the mean square. The proposed filter produces good results
given measurements supplied by low-cost inertial measurement and vision
units. Experiments with real-world data obtained from quadrotor have
demonstrated the strong tracking capabilities of the proposed filter
as it rapidly and accurately estimates the unknown pose and linear
velocity in presence of large initial attitude, position, and linear
velocity error and high level of measurement uncertainties. 

\section*{Acknowledgment}
The authors would like to thank \textbf{Maria Shaposhnikova} for proofreading
the article.

\subsection*{Appendix A\label{subsec:Appendix-A}}
\begin{center}
	\textbf{Quaternion of the Proposed Observers}
	\par\end{center}

\noindent Consider $Q=[q_{0},q^{\top}]^{\top}\in\mathbb{S}^{3}$ as
a unit-quaternion vector where $\mathbb{S}^{3}=\{\left.Q\in\mathbb{R}^{4}\right|||Q||=\sqrt{q_{0}^{2}+q^{\top}q}=1\}$
with $q_{0}\in\mathbb{R}$ and $q\in\mathbb{R}^{3}$. Let $\hat{Q}\in\mathbb{S}^{3}$
be the estimate of $Q\in\mathbb{S}^{3}$ and consider the mapping
from $\mathbb{S}^{3}$ to $\mathbb{SO}(3)$ to be $\mathcal{R}_{\hat{Q}}=(\hat{q}_{0}^{2}-||\hat{q}||^{2})\mathbf{I}_{3}+2\hat{q}\hat{q}^{\top}+2\hat{q}_{0}\left[\hat{q}\right]_{\times}\in\mathbb{SO}\left(3\right)$.
Recall the vector measurements in \eqref{eq:NAV_Set_Measurements}:

\noindent 
\begin{equation}
\begin{cases}
p_{c} & =\frac{1}{s_{T}}\sum_{i=1}^{n}s_{i}p_{i},\hspace{1em}s_{T}=\sum_{i=1}^{n}s_{i}\\
M & =\sum_{i=1}^{n}s_{i}p_{i}p_{i}^{\top}-s_{T}p_{c}p_{c}^{\top}\\
M\tilde{R} & =\sum_{i=1}^{n}s_{i}(p_{i}-p_{c})y_{i}^{\top}\mathcal{R}_{\hat{Q}}^{\top}\\
\tilde{R}^{\top}\tilde{P}_{\varepsilon} & =\frac{1}{s_{T}}\sum_{i=1}^{n}s_{i}(p_{i}-\mathcal{R}_{\hat{Q}}y_{i}-\hat{P})
\end{cases}\label{eq:NAV_Set_MeasurementsQ}
\end{equation}
where $\boldsymbol{\Upsilon}(M\tilde{R})=\mathbf{vex}(\boldsymbol{\mathcal{P}}_{a}(M\tilde{R}))$
and $||M\tilde{R}||_{{\rm I}}=\frac{1}{4}{\rm Tr}\{M(\mathbf{I}_{3}-\tilde{R})\}$.
Define $[e_{1},e_{2},e_{3},e_{4}]^{\top}=\left[||M\tilde{R}||_{{\rm I}},\tilde{P}_{\varepsilon}^{\top}\tilde{R}\right]^{\top}$such
that
\begin{equation}
\begin{cases}
E_{i} & =\frac{1}{2}\text{ln}\frac{\delta_{i}+e_{i}[k]/\xi_{i}[k]}{\delta_{i}-e_{i}[k]/\xi_{i}[k]}\\
\Delta_{i} & =\frac{1}{2\xi_{i}[k]}(\frac{1}{\delta_{i}+e_{i}[k]/\xi_{i}[k]}+\frac{1}{\delta_{i}-e_{i}[k]/\xi_{i}[k]})
\end{cases}\label{eq:PPF_Q}
\end{equation}
for all $i=1,\ldots,4$. Let $E=\left[E_{R},E_{P}^{\top}\right]^{\top}$,
$\Delta_{R}=\Delta_{1}$, and $\Delta_{P}={\rm diag}(\Delta_{2},\Delta_{3},\Delta_{4})$.
The equivalent quaternion-based navigation filter in \eqref{eq:NAV_Filter1_Detailed}
is given below:

\begin{equation}
\begin{cases}
\Theta_{m}= & \left[\begin{array}{cc}
0 & -\Omega_{m}^{\top}\\
\Omega_{m} & -[\Omega_{m}]_{\times}
\end{array}\right],\hspace{1em}\Psi=\left[\begin{array}{cc}
0 & -w_{\Omega}^{\top}\\
w_{\Omega} & [w_{\Omega}]_{\times}
\end{array}\right]\\
\dot{\hat{Q}} & =\frac{1}{2}\Theta_{m}\hat{Q}-\frac{1}{2}\Psi\hat{Q}\\
\dot{\hat{P}} & =\hat{V}-\left[w_{\Omega}\right]_{\times}\hat{P}-w_{V}\\
\dot{\hat{V}} & =\mathcal{R}_{\hat{Q}}a_{m}+\overrightarrow{\mathtt{g}}-\left[w_{\Omega}\right]_{\times}\hat{V}-w_{a}\\
w_{\Omega} & =-k_{w}(\Delta_{R}E_{R}+1)\boldsymbol{\Upsilon}(M\tilde{R})\\
& \hspace{1em}-\frac{\Delta_{R}}{4}\frac{||M\tilde{R}||_{{\rm I}}+2}{||M\tilde{R}||_{{\rm I}}+1}\hat{R}{\rm diag}(\hat{R}^{\top}\boldsymbol{\Upsilon}(M\tilde{R}))\hat{\sigma}\\
w_{V} & =\left[p_{c}\right]_{\times}w_{\Omega}-\frac{k_{v}}{\varepsilon}\Delta_{P}E_{P}-\ell_{P}\tilde{R}^{\top}\tilde{P}_{\varepsilon}\\
w_{a} & =-\overrightarrow{\mathtt{g}}-k_{a}\left(\frac{k_{v}}{\mu}\Delta_{P}+\mathbf{I}_{3}\right)\Delta_{P}E_{P}\\
k_{R} & =\gamma_{\sigma}\frac{||M\tilde{R}||_{{\rm I}}+2}{8}\Delta_{R}^{2}\exp(E_{R})\\
\dot{\hat{\sigma}}_{\Omega} & =k_{R}{\rm diag}(\hat{R}^{\top}\boldsymbol{\Upsilon}(M\tilde{R}))\hat{R}^{\top}\boldsymbol{\Upsilon}(M\tilde{R})-k_{\sigma}\gamma_{\sigma}\hat{\sigma}
\end{cases}\label{eq:NAV_Filter1_Detailed_Q}
\end{equation}

{\small
\rule{0.47\textwidth}{1pt}\\
Bibtex citation: \\
@article\{hashim2021Geom, \\
title=\{Geometric Stochastic Filter with Guaranteed Performance for
Autonomous Navigation based on IMU and Feature Sensor Fusion\}, \\
author=\{Hashim, Hashim A and Abouheaf, Mohammed and Abido, Mohammad
A\}, \\
journal=\{Control Engineering Practice\}, \\
volume=\{116\}, \\
pages=\{104926\}, \\
year=\{2021\} \\
\}\\
doi: \href{https://doi.org/10.1016/j.conengprac.2021.104926}{10.1016/j.conengprac.2021.104926}\\
Video URL: \href{https://youtu.be/ISUsnQbvz74}{youtu.be/ISUsnQbvz74}\\
\rule{0.49\textwidth}{1pt}
}

\vspace{10pt}

\bibliographystyle{IEEEtran}
\bibliography{bib_Navigation}

\end{document}